\documentclass[preprint,3p,number]{elsarticle}

\journal{arXiv}

\usepackage{float}
\floatstyle{boxed}
\restylefloat{figure}

\usepackage{graphicx} 
\usepackage{fontawesome}
\usepackage{proof}
\usepackage{amsmath}
\usepackage{amsthm}
\usepackage{amssymb}
\usepackage{stmaryrd}
\SetSymbolFont{stmry}{bold}{U}{cmr}{m}{n} 
\usepackage{hyperref}
\hypersetup{
  colorlinks=true,
  linkcolor=blue,
  anchorcolor=blue,
  citecolor=blue,
  urlcolor=blue,
  filecolor=blue,
  breaklinks=true
}
\usepackage[T1]{fontenc}
\usepackage[utf8]{inputenc}
\allowdisplaybreaks

\newtheorem{theorem}{Theorem}[section]
\newtheorem{proposition}[theorem]{Proposition}
\newtheorem{corollary}[theorem]{Corollary}
\newtheorem{definition}[theorem]{Definition}
\newtheorem{remark}[theorem]{Remark}
\newtheorem{example}[theorem]{Example}

\newcommand\faex{{\rotatebox[origin=c]{180}{{\sffamily\AE}}}}
\newcommand\ket[1]{\ensuremath{|{#1}\rangle}}
\newcommand\abstr[1]{#1.}
\newcommand\B{{\mathcal B}}
\newcommand\Q{{\mathcal Q}}
\newcommand\inl{\mbox{\it inl}}
\newcommand\inr{\mbox{\it inr}}
\newcommand\elimtop{\delta_{\top}}
\newcommand\elimbot{\delta_{\bot}}
\newcommand\elimand{\delta_{\wedge}}
\newcommand\elimor{\delta_{\vee}}
\newcommand\elimsup{\delta_{\odot}}
\newcommand\super[2]{[#1,#2]}
\newcommand\pair[2]{\langle #1, #2 \rangle}
\newcommand\test{\mbox{\it If}}
\newcommand\lra{\longrightarrow}

\newcommand\irule[3]{\infer[\mbox{\footnotesize $#3$}]{#2}{#1}}
\newcommand\irulePi[3]{\infer*[\mbox{\footnotesize $#3$}]{#2}{#1}}
\newcommand{\plus}{\mathbin{\text{\normalfont\scalebox{0.7}{\faPlus}}}}

\begin{document}
\begin{frontmatter}
\title{A New Connective in Natural Deduction,\texorpdfstring{\\}{ }and its Application to Quantum Computing\texorpdfstring{\tnoteref{thanks}}{}}
\tnotetext[thanks]{Funded by STIC-AmSud 21STIC10, PIP 11220200100368CO, PICT-2019-1272, PICT-2021-I-A-00090, and the French-Argentinian IRP SINFIN.}

\author[UNQ,ICC]{Alejandro D\'iaz-Caro}
\ead{adiazcaro@icc.fcen.uba.ar}
\author[Inria,ENS]{Gilles Dowek}
\ead{gilles.dowek@ens-paris-saclay.fr}

\address[UNQ]{Departamento de Ciencia y Tecnolog\'ia, Universidad Nacional de Quilmes, Argentina}
\address[ICC]{Instituto de Ciencias de la Computaci\'on, CONICET--Universidad de Buenos Aires, Argentina}
\address[Inria]{Inria, France}
\address[ENS]{\'Ecole Normale Sup\'erieure Paris-Saclay, France}

\begin{abstract}
We investigate an unsuspected connection between logical connectives
with non-harmonious deduction rules, such as Prior's {\it tonk}, and
quantum computing. We argue that these connectives model the
information-erasure, the non-reversibility, and the non-determinism
that occur, among other places, in quantum measurement.  We introduce
an intuitionistic propositional logic with a non-harmonious logical
connective {\it sup}\/ and two interstitial rules, and show that the
proof language of this logic forms the core of a quantum programming
language.
\end{abstract}

  \begin{keyword}
    Proof-reduction
    \sep
    Lambda calculus
    \sep
    Type theory
    \sep
    Quantum computing
\MSC[2020] 03F05 \sep 68N18 \sep 03B38 \sep 03B70  \sep 81P68 
\end{keyword}

\end{frontmatter}

\section{Introduction}

A puzzling question in quantum physics is whether, in a quantum
superposition $\alpha \ket{\phi} + \beta \ket{\psi}$ of two states
$\ket \phi$ and $\ket \psi$, there is the state $\ket \phi$ {\it and}
the state $\ket \psi$ or the state $\ket \phi$ {\it or}\/ the state
$\ket \psi$.

Indeed, when we build such a superposition, that is when we prepare
it, we need to have $\ket \phi$ {\it and} $\ket \psi$, but when we use
this state, that is when we measure it, we get $\ket \phi$ {\it or}
$\ket \psi$. Thus, this superposition is similar to a conjunction when
we build it, but to a disjunction when we use it.  This discrepancy
between the way such a superposition is built and used is reminiscent
of the natural deduction rules of the non-harmonious connectives such
as Prior's {\it tonk}, and others.  We defend, in this paper, the
thesis that these non-harmonious connectives model the
information-erasure, the non-reversibility, and the non-determinism
that occur, among other places, in quantum measurement, while the
harmonious ones model information preservation, reversibility, and
determinism.

More concretely, after discussing the notion of harmonious and
non-harmonious deduction rules (Section~\ref{harmony}), we introduce
an intuitionistic propositional logic with a logical connective
$\odot$ (read: ``sup'', for ``superposition'') that has non-harmonious
deduction rules, we introduce a language of proof-terms for this
logic, the $\odot$-calculus (read: ``the sup-calculus''), and we prove
its main properties: subject reduction, the termination of
proof-reduction, the introduction property, and partial confluence
(Section~\ref{seclogic}).  These proofs mostly use standard techniques
with some specificities, to be adapted to this calculus.  We then
extend this calculus, introducing scalars to quantify the propensity
of a proof to reduce to another (Section~\ref{secquantifying}) and
show that this proof language contains the core of a quantum
programming language (Section~\ref{secquantum}).  Note that that
intuitionistic propositional logic with $\odot$ is not a logic to
reason about quantum programs. It is a logic whose propositions are
the types of quantum programs.

The idea to extend logic with a new connective for quantum
superposition has already been investigated, for example in
\cite{suplogic17,suplogic19}. While these papers address the question
of the models of such a logic, we address that of the dynamic of its
proofs.

A preliminary version of this paper has been published in the
proceedings of the {\it International Colloquium on Theoretical
  Aspects of Computing}, 2021. In this journal version, we have
replaced the symbol $\delta_{\odot}^{\parallel}$ with the symbols
$\elimsup^1$, $\elimsup^2$, clarifying the two-face nature of the
connective $\odot$.  We have also introduced an elimination rule for
the symbol $\top$. Such a rule is often considered as redundant, but
it fully makes sense in natural deduction with generalized elimination
rules, and even more in a proof system with scalars, such as that of
Section~\ref{secquantifying}.  Besides providing the complete proofs
of all theorems, we investigate, in this paper, the confluence of the
deterministic part of the calculus, that was not addressed in the
conference version.  To our surprise, the system without scalars was
confluent, but the system with scalars was not. This led us to modify
the treatment of scalars and the definition of matrices to make this
system confluent. Finally, the conference version of the paper only
addressed quantum algorithms on one and two qubits. In this version,
we have generalized this to arbitrary quantum algorithms, leading to a
more systematic treatment of vectors, matrices, and measurement.

\section{Harmony and excessiveness}
\label{harmony}

\subsection{Logical connectives with insufficient, harmonious, and excessive
  deduction rules}

In natural deduction, to prove a proposition of the form $A
\vartriangle B$, where $\vartriangle$ is an arbitrary connective, the
introduction rules of this connective require to prove some
propositions.  Then, when we have a proof of $A \vartriangle B$ and we
want to prove a proposition $C$, the generalized elimination rules
\cite{SchroederHeister,Dyckhoff,Parigot,LopezEscobar,Plato,NegriPlato}
of the connective $\vartriangle$ provide some hypotheses for proving
this proposition. In general, the propositions required by the
introduction rules and those provided by the elimination rules are the
same.

For example, to prove the proposition $A \wedge B$, the introduction
rule of the conjunction 
$$\irule{\Gamma \vdash A & \Gamma \vdash B}
        {\Gamma \vdash A \wedge B}
        {\mbox{$\wedge$-i}}$$
requires proofs of $A$ and of $B$ and, to prove a proposition $C$, the
generalized elimination rules of the conjunction
$$\irule{\Gamma \vdash A \wedge B & \Gamma, A \vdash C}
        {\Gamma \vdash C}
        {\mbox{$\wedge$-e1}}
\qquad
\irule{\Gamma \vdash A \wedge B & \Gamma, B \vdash C}
        {\Gamma \vdash C}
        {\mbox{$\wedge$-e2}}$$
provide the same propositions $A$ and $B$.
This 
principle of inversion, or of harmony, has been introduced by Gentzen
\cite{Gentzen} and developed, among others, by Prawitz \cite{Prawitz,PrawitzEssay}, Dummett \cite{Dummett}, and Schoeder-Heister \cite{SchroederHeister,SchroederHeister2014} for natural deduction, by Miller and Pimentel \cite{MillerPimentel} for sequent calculus, and by Read \cite{Read04,Read10,Read} and Martin-L\"of \cite{MartinLof} for the rules of equality.
It enables the definition of
a reduction process, where the proof
$$\irule{\irule{\irulePi{\pi_1}{\Gamma \vdash A}{}
                &
                \irulePi{\pi_2}{\Gamma \vdash B}{}
               }
               {\Gamma \vdash A \wedge B}
               {\mbox{$\wedge$-i}}
         &
         \irulePi{\pi_3}{\Gamma, A \vdash C}{}
        }
        {\Gamma \vdash C}
        {\mbox{$\wedge$-e1}}$$
reduces to $(\pi_1/A)\pi_3$, that is the proof $\pi_3$ where
every occurrence of the rule axiom with the proposition $A$ has been
replaced with the proof $\pi_1$. And, similarly, the proof 
$$\irule{\irule{\irulePi{\pi_1}{\Gamma \vdash A}{}
                &
                \irulePi{\pi_2}{\Gamma \vdash B}{}
               }
               {\Gamma \vdash A \wedge B}
               {\mbox{$\wedge$-i}}
         &
         \irulePi{\pi_3}{\Gamma, B \vdash C}{}
        }
        {\Gamma \vdash C}
        {\mbox{$\wedge$-e2}}$$
reduces to $(\pi_2/B)\pi_3$.

This property that the elimination rules provide exactly the
propositions required by the introduction rules can be split into two
properties, that it provides no more and no less (called ``harmony''
and ``stability'' in \cite{JacintoReadSL17}).

It is of course possible to imagine other deduction rules that do not
respect this harmony principle, either because the introduction rules
do not require some propositions, that are provided by the
elimination rules (such rules may be called {\it insufficient}\/)
$$\irule{\Gamma \vdash A}
        {\Gamma \vdash A \smile B}
        {\mbox{$\smile$-i}}$$
$$\irule{\Gamma \vdash A \smile B & \Gamma, A, B \vdash C}
        {\Gamma \vdash C}
        {\mbox{$\smile$-e}}$$
or because the introduction rules require some propositions, that are
not provided by the elimination rules (such rules may be called {\em
  excessive}\/)
$$\irule{\Gamma \vdash A & \Gamma \vdash B}
        {\Gamma \vdash A \frown B}
        {\mbox{$\frown$-i}}$$
$$\irule{\Gamma \vdash A \frown B & \Gamma, A \vdash C}
        {\Gamma \vdash C}
        {\mbox{$\frown$-e}}$$
or both, such as Prior's {\it tonk} \cite{Prior} 
$$\irule{\Gamma \vdash A}
        {\Gamma \vdash A~\mbox{\it tonk}~B}
        {\mbox{{\it tonk}-i}}$$
$$\irule{\Gamma \vdash A~\mbox{\it tonk}~B & \Gamma, B \vdash C}
        {\Gamma \vdash C}
        {\mbox{{\it tonk}-e}}$$

More generally, we can associate to each set of introduction rules a
{\it required}\/ proposition. When an introduction rule has several
premisses, we take the conjunction of these premises, and when there
are several such introduction rules, we take the disjunction of their
premises.  In the same way, we can associate to each set of
elimination rules a {\it provided}\/ proposition. As the provided
propositions occur on the left of the sequent, when an elimination
rule has several premisses, we take their disjunction,
and when there are several such elimination rules, we take the
conjunction of their premises.
$$\begin{array}{|c|c|c|}
	  \hline
	  & \mbox{required} & \mbox{provided} \\
	  \hline
	  \wedge & \mbox{$A$ and $B$}  & \mbox{$A$ and $B$}\\
	  \hline
	  \vee & \mbox{$A$ or $B$} & \mbox{$A$ or $B$}\\
	  \hline
	  \smile & \mbox{$A$} & \mbox{$A$ and $B$}\\
	  \hline
	  \frown & \mbox{$A$ and $B$} & \mbox{$A$}\\
	  \hline
	  \mbox{\it tonk} & \mbox{$A$} & \mbox{$B$}\\
	  \hline
\end{array}$$
A set of rules is {\it harmonious}\/ when the required and the provided
propositions are the equivalent. It is {\it excessive}\/ when the required
proposition implies the provided proposition, but not the converse. It is
{\it insufficient}\/ when the provided proposition implies the required one,
but not the converse.  The connective $\wedge$ and $\vee$ are thus
harmonious. The connective $\smile$ is insufficient. The connective $\frown$
is excessive. And, as both $A \Rightarrow B$ and $B \Rightarrow A$ are
unprovable, the connective {\it tonk}\/ is neither.

When a set of rule is excessive, a proof reduction process can still
be defined, for example the proof
$$\irule{\irule{\irulePi{\pi_1}{\Gamma \vdash A}{}
                &
                \irulePi{\pi_2}{\Gamma \vdash B}{}
               }
               {\Gamma \vdash A \frown B}
               {\mbox{$\frown$-i}}
          &
               \irulePi{\pi_3}{\Gamma, A \vdash C}{}
        }
        {\Gamma \vdash C}
        {\mbox{$\frown$-e}}$$
can still be reduced to $(\pi_1/A)\pi_3$.

Another example is the connective $\odot$ whose introduction rule 
$$\irule{\Gamma \vdash A & \Gamma \vdash B}
        {\Gamma \vdash A \odot B}
        {\mbox{$\odot$-i}}$$
similar to that of the conjunction, requires $A$ and $B$ and whose
elimination rule
$$\irule{\Gamma \vdash A \odot B & \Gamma, A \vdash C & \Gamma, B \vdash C}
        {\Gamma \vdash C}
        {\mbox{$\odot$-e}}$$
similar to that of the disjunction, provides $A$ or $B$.
In this case also, proofs can be reduced. Moreover, several proof reduction
processes can be defined, exploiting, in different ways,
the excess of the deduction rules. For example, the proof
$$\irule{\irule{\irulePi{\pi_1}{\Gamma \vdash A}{}
                &
                \irulePi{\pi_2}{\Gamma \vdash B}{}
               }
               {\Gamma \vdash A \odot B}
               {\mbox{$\odot$-i}}
          &
          \irulePi{\pi_3}
                {\Gamma, A \vdash C}{} & \irulePi{\pi_4}{\Gamma, B \vdash C}
                {}
        }
        {\Gamma \vdash C}
        {\mbox{$\odot$-e}}$$
can be reduced to $(\pi_1/A)\pi_3$, it can be reduced to
$(\pi_2/B)\pi_4$, it also can be reduced,
in a non-deterministic way, 
either to $(\pi_1/A)\pi_3$ or to $(\pi_2/B)\pi_4$.

This notion of harmony can also be extended to the rules of
quantifiers.  Then, the quantifier $\faex$, whose introduction rule
$$\irule{\Gamma \vdash A}
        {\Gamma \vdash \faex x~A}
        {\mbox{$\faex$-i $x$ not free in $\Gamma$}}$$
similar to that of the universal quantifier, requires a proof of $A$
for all $x$ and whose elimination rule
$$\irule{\Gamma \vdash \faex x~A & \Gamma, A \vdash C}
        {\Gamma \vdash C}
        {\mbox{$\faex$-e $x$ not free in $\Gamma, C$}}$$
similar to that of the existential quantifier, provides a proof of $A$
for some $x$, has excessive deduction rules.  The proof
$$\irule{\irule{\irulePi{\pi_1}{\Gamma \vdash A}{}
               }
               {\Gamma \vdash \faex x~A}
               {\mbox{$\faex$-i}}
          &
               \irulePi{\pi_2}{\Gamma, A \vdash C}{}
        }
        {\Gamma \vdash C}
        {\mbox{$\faex$-e}}$$
can be reduced, in a non-deterministic way, to
$((t/x)\pi_1/A)(t/x)\pi_2$, for any term $t$.

The quantifier $\nabla$ \cite{MillerTiu}, defined in sequent calculus
rather than natural deduction, may also be considered as a quantifier
with excessive deduction rules, as it has the right rule of the
universal quantifier and the left rule of the existential one. But it
involves a clever management of variable scoping, which we do not
address here.

\subsection{Mixing excessiveness and harmony}

The rules
$$\irule{\Gamma \vdash A & \Gamma \vdash B}
        {\Gamma \vdash A \odot B}
        {\mbox{$\odot$-i}}$$

$$\irule{\Gamma \vdash A \odot B & \Gamma, A \vdash C & \Gamma, B \vdash C}
        {\Gamma \vdash C}
        {\mbox{$\odot$-e}}$$
are excessive. 

But, we can add another set of elimination rules for the connective $\odot$,
similar to those of conjunction
$$\irule{\Gamma \vdash A\odot B & \Gamma, A \vdash C}
        {\Gamma \vdash C}
        {\mbox{$\odot$-e1}}$$
$$\irule{\Gamma \vdash A\odot B & \Gamma, B \vdash C}
        {\Gamma \vdash C}
        {\mbox{$\odot$-e2}}$$
Then, the connective $\odot$, with its four rules $\odot$-i,
$\odot$-e, $\odot$-e1, and $\odot$-e2, appears as a two-face connective:
the subset of its deduction rules $\{\mbox{$\odot$-i},
\mbox{$\odot$-e}\}$ is excessive, while the subset
$\{\mbox{$\odot$-i}, \mbox{$\odot$-e1}, \mbox{$\odot$-e2}\}$ is
harmonious.
Note that the rules $\{\mbox{$\odot$-i}, \mbox{$\odot$-e1},
\mbox{$\odot$-e2}\}$ are exactly those of the conjunction.

\subsection{Information loss}

We say that an occurrence of a sub-proof $\pi_1$ of a proof $\pi$ is
{\it accessible}, if there exists a context $\kappa$ such that
$\kappa\{\pi_X\}$, where $\pi_X$ is obtained by replacing this
occurrence of $\pi_1$ with a variable $X$, reduces to $X$.

For example, the occurrence of $\pi_1$ in the proof 
$$\irule{\irulePi{\pi_1}{\Gamma \vdash A}{}
          &
         \irulePi{\pi_2}{\Gamma \vdash B}{}
        }
        {\Gamma \vdash A \wedge B}
        {\mbox{$\wedge$-i}}$$
is accessible, as putting the proof
$$\irule{\irulePi{X}{\Gamma \vdash A}{}
          &
         \irulePi{\pi_2}{\Gamma \vdash B}{}
        }
        {\Gamma \vdash A \wedge B}
	{\mbox{$\wedge$-i}}$$
in the context
$$\irule{\irulePi{\{\}}
               {\Gamma \vdash A \wedge B}
               {}
         &
         \irule{}{\Gamma, A \vdash A}{\mbox{axiom}}
        }
        {\Gamma \vdash A}
	{\mbox{$\wedge$-e1}}$$
yields the proof
$$\irule{\irule{\irulePi{X}{\Gamma \vdash A}{}
                &
                \irulePi{\pi_2}{\Gamma \vdash B}{}
               }
               {\Gamma \vdash A \wedge B}
               {\mbox{$\wedge$-i}}
         &
         \irule{}{\Gamma, A \vdash A}{\mbox{axiom}}
        }
        {\Gamma \vdash A}
        {\mbox{$\wedge$-e1}}$$
that reduces to $X$.
In other words, the rule $\wedge$-i puts the proofs $\pi_1$ and $\pi_2$
in a box, but the box can be opened and the proofs can be taken out
of it.

With harmonious deduction rules, when a proof is built with an
introduction rule, the proofs of its premises remain accessible.
We call this property {\it information preservation}. 
The
situation is different with excessive deduction rules: the excess of
information, required by the introduction rule, and not returned by
the elimination rule in the form of an extra hypothesis in the
required proof of $C$ is lost.  For example, the occurrence of $\pi_2$
in the proof
$$\irule{\irulePi{\pi_1}{\Gamma \vdash A}{}
              &
         \irulePi{\pi_2}{\Gamma \vdash B}{}
        }
        {\Gamma \vdash A \frown B}
        {\mbox{$\frown$-i}}$$
is inaccessible as there is no context such that putting the proof 
$$\irule{\irulePi{\pi_1}{\Gamma \vdash A}{}
              &
         \irulePi{X}{\Gamma \vdash B}{}
        }
        {\Gamma \vdash A \frown B}
        {\mbox{$\frown$-i}}$$
in that context yields a proof that reduces to $X$.
Again, the rule $\frown$-i puts the proofs $\pi_1$ and $\pi_2$ in
a box, the box can be partially opened and the proof $\pi_1$ can be
taken out of it, but not the proof $\pi_2$, that is inaccessible.
The information it contains is lost.

The accessibility of the occurrences of $\pi_1$ and $\pi_2$ of the proof 
$$\irule{\irulePi{\pi_1}{\Gamma \vdash A}{}
           & 
         \irulePi{\pi_2}{\Gamma \vdash B}{}
        }
        {\Gamma \vdash A \odot B}
        {\mbox{$\odot$-i}}$$
depends on the elimination rules we allow in the context and on the
way we reduce the proof
$$\irule{\irule{\irulePi{\pi_1}{\Gamma \vdash A}{}
    &
    \irulePi{\pi_2}{\Gamma \vdash B}{}
  }
  {\Gamma \vdash A \odot B}
  {\mbox{$\odot$-i}}
  &
  \irulePi{\pi_3}
  {\Gamma, A \vdash C}{} & \irulePi{\pi_4}{\Gamma, B \vdash C}
  {}
}
{\Gamma \vdash C}
{\mbox{$\odot$-e}}$$
If we allow the rule $\odot$-e in the context, but neither $\odot$-e1
nor $\odot$-e2, and reduce this proof systematically to
$(\pi_1/A)\pi_3$, then $\pi_1$ is accessible, but $\pi_2$ is not.  If
we reduce it systematically to $(\pi_2/B)\pi_4$, then $\pi_2$ is
accessible, but $\pi_1$ is not.  If we reduce it,
in a non-deterministic way,
to $(\pi_1/A)\pi_3$ or to $(\pi_2/B)\pi_4$, then both $\pi_1$ and
$\pi_2$ are accessible, but in a non-deterministic way.  If we allow
the rule $\odot$-e1 and $\odot$-e2 in the context, then both proofs
are accessible.
Once more, the rule $\odot$-i puts the proofs $\pi_1$ and $\pi_2$ in
a box, whether the box can be opened and the proofs
taken out of it, depends on the tools we use to open it.

When a connective has non-harmonious deduction rules, its introduction
rules alone do not define its meaning, and neither do the elimination
rules alone. The discrepancy between the meaning conferred by the
introduction rules and the elimination rules, and the information loss
it implies, are part of the meaning of such a connective.

While connectives with harmonious deduction rules model
information-preservation, reversibility, and determinism, those with
excessive deduction rules model information-erasure,
non-reversibility, and non-determinism.  In particular, the
elimination rules $\odot$-e1 and $\odot$-e2 will be used to model
information-preservation, reversibility, and determinism, while the
elimination rule $\odot$-e will be used to model information-erasure,
non-reversibility, and non-determinism.

Such information-erasure, non-reversibility, and non-determinism,
occur, for example, in quantum physics, where the measurement of the
superposition of two states does not yield both states back.

\subsection{Quantum computing and quantum languages}

In classical computing, the simplest datatype is that of bits, that
contains two elements $0$ and $1$.  In quantum computing, (see
\cite{NC} for a more comprehensive introduction) this datatype is
replaced with that of qubits, that contains all linear combinations
$\alpha \ket 0 + \beta \ket 1$, of $0$ and $1$, then written $\ket 0$
and $\ket 1$, with complex coefficients $\alpha$ and $\beta$ such that
$|\alpha|^2 + |\beta|^2 = 1$. Thus, this datatype of qubits contains
the vectors of norm $1$ of the vector space $\mathbb C^2$.

In classical computing, a more complex datatype is that of $n$-bits.
The elements of this datatype are $n$-tuples formed with $0$ and $1$,
for instance the datatype of $3$-bits contains the $8$ elements $000$,
$001$, $010$, $011$, $100$, $101$, $110$, and $111$.  In quantum
computing, this datatype is replaced with that of $n$-qubits, that
contains all linear combinations of norm $1$ of these tuples.  The
elements of this datatype are thus vectors of norm $1$ of $\mathbb
C^{2^n}$.  The canonical basis of $\mathbb C^{2^n}$ is made of the
$2^n$ tuples formed with $0$ and $1$. For instance the canonical basis
of $\mathbb C^8$ is $\{\ket{000}, \ket{001}, \ket{010}, \ket{011},
\ket{100}, \ket{101}, \ket{110}, \ket{111}\}$.  These vectors can also
be written $\{\ket{0}, \ket{1}, \ket{2}, \ket{3}, \ket{4}, \ket{5},
\ket{6}, \ket{7}\}$ identifying a natural number with its binary
notation.  More generally, the datataypes contain the vectors of norm
$1$ of some Hilbert space.

Quantum algorithms are formed with two ingredients: unitary
transformations, that are linear transformations preserving inner
product, and measurements, that are probabilistic operations. It is in
fact enough to consider the projection to the canonical basis which
acts as follows: the measurement of an arbitrary qubit
$\ket\psi=\alpha\ket 0+\beta\ket 1$ yields the result $0$ and
transforms the qubit into $\ket 0$ with probability $|\alpha|^2$ and
yields the result $1$ and transform the qubit into $\ket 1$ with
probability $|\beta|^2$. More generally, the measurement of the first
qubit of the $n$-qubit $\sum_{i=0}^{2^n-1}\alpha_i\ket i$ yields $0$
and transforms it into
$$\sum_{i=0}^{2^{n-1}-1}\frac{\alpha_i}{\|\sum_{i=0}^{2^{n-1}-1}\alpha_i\ket
  i\|}\ket i$$
with probability $\sum_{i=0}^{2^{n-1}-1}|{\alpha_i}|^2$, and yields
$1$ and transforms it into
$$\sum_{i=2^{n-1}}^{2^{n}-1}\frac{\alpha_i}{\|\sum_{i=2^{n-1}}^{2^{n}-1}\alpha_i\ket
  i\|}\ket i$$
with probability $\sum_{i=2^{n-1}}^{2^{n}-1}|{\alpha_i}|^2$.

Measuring any other of the $n$-qubits involves the
swapping unitary transformation and consist into swapping each qubit
to measure with the first one, perform the measurement, and swapping
it back.

Several programming languages have been proposed to express quantum
algorithms, for example
\cite{AltenkirchGrattageLICS05,SelingerValironMSCS06,ZorziMSCS16,Lineal,ZXBook17,LambdaS,DiazcaroGuillermoMiquelValironLICS19}.
The design of such quantum programming languages raises two main
questions. The first is to take into account the linearity of the
unitary operators and, for instance, avoid cloning, and the second is
to express the information-erasure, non-reversibility, and
non-determinism of the measurement. The $\odot$ connective gives a new
solution to this second problem. Qubits can be seen as proofs of the
proposition $\top \odot \top$, in contrast with bits which are proofs
of $\top\vee\top$, and measurement can be easily expressed with the
elimination rule $\odot$-e (Section~\ref{measurement}).

In previous works, we have attempted to formalize superposition and
measurement in the $\lambda$-calculus.  The calculus Lambda-${\mathcal
  S}$~\cite{LambdaS} contains a primitive constructor $+$ and a
primitive measurement symbol $\pi$, together with a rule reducing $\pi
(t + u)$, in a non-deterministic way, to $t$ or to $u$.
The superposition $t + u$ can also 
be considered as the pair $\pair{t}{u}$.  Hence, it should have the
type $A \wedge A$. In other words, it should be a proof of the
proposition $A \wedge A$. In System I \cite{SystemI}, various
type-isomorphisms have been taken as identities, in particular the
commutativity isomorphism $A \wedge B \equiv B \wedge A$, hence $t + u
\equiv u + t$. In such a system, where $A \wedge B$ and $B \wedge A$
are identical, it is not possible to define the two elimination rules
as the two usual projections rules $\pi_1$ and $\pi_2$ of the
$\lambda$-calculus. They were replaced with a single projection
parametrized with a proposition $A$: $\pi_A$, such that if $t:A$ and
$u:B$ then $\pi_A(t + u)$ reduces to $t$ and $\pi_B(t + u)$ to $u$.
When $A = B$, hence $t$ and $u$ both have type $A$, the proof-term
$\pi_A(t + u)$ reduces,
in a non-deterministic way,
to $t$ or to $u$, like 
a measurement operator.

These works on Lambda-${\mathcal S}$ and System I brought to light the
fact that the pair superposition / measurement, in a quantum
programming language, behaves like a pair introduction / elimination,
for some connective, in a proof language, as the succession of a
superposition and a measurement yields a term that can be reduced.  In
System I, this connective was assumed to be a commutative conjunction,
with a modified elimination rule, leading to a non-deterministic
reduction.

But, as the measurement of the superposition of two states does not
yield both states back, this connective should probably be
excessive. Moreover, as, to prepare the superposition $\alpha.\ket 0 + \beta.\ket 1$,
we need both $\ket 0$ and $\ket 1$ and the measurement in
the basis $\ket 0, \ket 1$ yields either $\ket 0$ or $\ket 1$, this
connective should have the introduction rule of the conjunction, and
the elimination rule of the disjunction. Hence, it should be the
connective $\odot$.

\section{Intuitionistic propositional logic with \texorpdfstring{$\odot$}{O}}
    \label{seclogic}
\begin{figure}[t]
$$\irule{}
        {\Gamma \vdash A}
        {\mbox{axiom~$A\in\Gamma$}}
  \qquad
  \irule{\Gamma \vdash A & \Gamma \vdash A}
        {\Gamma \vdash A}
        {\mbox{sum}}
  \qquad
  \irule{\Gamma \vdash A}
        {\Gamma \vdash A}
        {\mbox{prod}}$$
$$\irule{}
        {\Gamma \vdash \top}
        {\mbox{$\top$-i}}
  \qquad
  \irule{\Gamma \vdash \top & \Gamma \vdash C}
        {\Gamma \vdash C}
        {\mbox{$\top$-e}}
  \qquad
  \irule{\Gamma \vdash \bot}
        {\Gamma \vdash C}
        {\mbox{$\bot$-e}}$$
$$\irule{\Gamma, A \vdash B}
        {\Gamma \vdash A \Rightarrow B}
        {\mbox{$\Rightarrow$-i}}
  \qquad\qquad
  \irule{\Gamma \vdash A\Rightarrow B & \Gamma \vdash A}
        {\Gamma \vdash B}
        {\mbox{$\Rightarrow$-e}}$$
$$\irule{\Gamma \vdash A & \Gamma \vdash B}
        {\Gamma \vdash A \wedge B}
        {\mbox{$\wedge$-i}}
  \qquad
  \irule{\Gamma \vdash A \wedge B & \Gamma, A \vdash C}
        {\Gamma \vdash C}
        {\mbox{$\wedge$-e1}}
  \qquad
  \irule{\Gamma \vdash A \wedge B & \Gamma, B \vdash C}
        {\Gamma \vdash C}
        {\mbox{$\wedge$-e2}}$$
$$\irule{\Gamma \vdash A}
        {\Gamma \vdash A \vee B}
        {\mbox{$\vee$-i1}}
  \qquad\quad
  \irule{\Gamma \vdash B}
        {\Gamma \vdash A \vee B}
        {\mbox{$\vee$-i2}}
  \qquad\quad
  \irule{\Gamma \vdash A \vee B & \Gamma, A \vdash C & \Gamma, B \vdash C}
        {\Gamma \vdash C}
        {\mbox{$\vee$-e}}$$
$$\irule{\Gamma \vdash A & \Gamma \vdash B}
        {\Gamma \vdash A \odot B}
        {\mbox{$\odot$-i}}
  \qquad\qquad
  \irule{\Gamma \vdash A \odot B & \Gamma, A \vdash C & \Gamma, B \vdash C}
        {\Gamma \vdash C}
        {\mbox{$\odot$-e}}$$
$$\irule{\Gamma \vdash A \odot B & \Gamma, A \vdash C}
        {\Gamma \vdash C}
        {\mbox{$\odot$-e1}}
  \qquad\qquad
  \irule{\Gamma \vdash A \odot B & \Gamma, B \vdash C}
        {\Gamma \vdash C}
        {\mbox{$\odot$-e2}}$$

\caption{The deduction rules of intuitionistic propositional 
          logic with $\odot$\label{figdeductionrules}}
\end{figure}

We consider an intuitionistic propositional logic with the usual
connectives $\top$, $\bot$, $\Rightarrow$, $\wedge$, and $\vee$
(as usual, negation is defined as $\neg A = (A \Rightarrow \bot)$),
and the extra connective $\odot$.
The syntax of the propositions of this logic is
$$A = \top \mid \bot \mid A \Rightarrow A \mid A \wedge A \mid A \vee
A \mid A \odot A$$
For
simplicity, we have not included propositional constants in this
syntax, so $\top \Rightarrow \top$ is a proposition, but $P
\Rightarrow P$ is not.  Adding such constants only requires minors
modifications in the proofs below.

The deduction rules are given in Figure~\ref{figdeductionrules}.  We use the
generalized elimination rules for all connectives, except implication, for
which the elimination rule is the usual {\it modus ponens}. This yields the
usual application in the language of proof-terms, making the examples more
readable.

The rules axiom, $\top$-i, $\top$-e,
$\bot$-e, $\Rightarrow$-i, $\Rightarrow$-e,
$\wedge$-i, $\wedge$-e1, $\wedge$-e2,
$\vee$-i1, $\vee$-i2, and $\vee$-e are the usual rules of 
intuitionistic propositional logic.
The rules $\odot$-i, $\odot$-e, $\odot$-e1, and $\odot$-e2 are those of
the connective $\odot$.
We also added two rules 
$$\irule{\Gamma \vdash A & \Gamma \vdash A}{\Gamma \vdash A}{\mbox{sum}}$$
$$\irule{\Gamma \vdash A} {\Gamma \vdash A} {\mbox{prod}}$$
whose premises are identical to their conclusion.  Although these
rules are logically trivial, they introduce constructors in the proof
language that will be of key importance when we extend the calculus
with scalars, in Section~\ref{secquantifying}.
These rules are called {\it interstitial}\/ as, as we shall see,
they can create an interstice
between the introduction rule of some connective and its
elimination rule.

\subsection{Proof reduction}

The reduction rules
in this logic are the usual ones for the
connectives $\Rightarrow$, $\wedge$, and $\vee$, except that we use the
generalized elimination rules for the conjunction.
$$\vcenter{\irule{\irule{\irulePi{\pi_1}{\Gamma, A\vdash B}{}}
               {\Gamma \vdash A \Rightarrow B}
               {\mbox{$\Rightarrow$-i}}
          &
               \irulePi{\pi_2}{\Gamma \vdash A}{}
        }
        {\Gamma \vdash B}
	{\mbox{$\Rightarrow$-e}}}
	\qquad
	\textrm{that reduces to}\qquad(\pi_2/A)\pi_1$$
$$\vcenter{\irule{\irule{\irulePi{\pi_1}{\Gamma \vdash A}{}
                 & 
                \irulePi{\pi_2}{\Gamma \vdash B}{}
               }
               {\Gamma \vdash A \wedge B}
               {\mbox{$\wedge$-i}}
         &
         \irulePi{\pi_3}{\Gamma, A \vdash C}{}
        }
        {\Gamma \vdash C}
	{\mbox{$\wedge$-e1}}}
	\qquad
	\textrm{that reduces to}\qquad(\pi_1/A)\pi_3$$
 $$\vcenter{\irule{\irule{\irulePi{\pi_1}{\Gamma \vdash A}{}
                 & 
                \irulePi{\pi_2}{\Gamma \vdash B}{}
               }
               {\Gamma \vdash A \wedge B}
               {\mbox{$\wedge$-i}}
         &
         \irulePi{\pi_3}{\Gamma, B \vdash C}{}
        }
        {\Gamma \vdash C}
	{\mbox{$\wedge$-e2}}}
	\qquad
	\textrm{that reduces to}\qquad(\pi_2/B)\pi_3$$
$$\vcenter{\irule{\irule{\irulePi{\pi_1}{\Gamma \vdash A}{}}
               {\Gamma \vdash A \vee B}
               {\mbox{$\vee$-i1}}
          &
               \irulePi{\pi_2}{\Gamma, A \vdash C}{}
               &
               \irulePi{\pi_3}{\Gamma, B \vdash C}{}
        }
        {\Gamma \vdash C}
	{\mbox{$\vee$-e}}}
	\qquad
	\textrm{that reduces to}\qquad(\pi_1/A)\pi_2$$
and
$$\vcenter{\irule{\irule{\irulePi{\pi_1}{\Gamma \vdash B}{}}
               {\Gamma \vdash A \vee B}
               {\mbox{$\vee$-i2}}
          &
               \irulePi{\pi_2}{\Gamma, A \vdash C}{}
               &
               \irulePi{\pi_3}{\Gamma, B \vdash C}{}
        }
        {\Gamma \vdash C}
	{\mbox{$\vee$-e}}}
	\qquad
	\textrm{that reduces to}\qquad(\pi_1/B)\pi_3$$
For the connective $\top$, we have decided to take an elimination rule. 
This introduces a notion of reducible proof for the connective $\top$        
$$\vcenter{\irule{\irule{}
                        {\Gamma \vdash \top}
                        {\mbox{$\top$-i}}
                  &
                  \irulePi{\pi}{\Gamma \vdash C}{}
                 }
                 {\Gamma \vdash C}
	         {\mbox{$\top$-e}}}
	\qquad
	\textrm{and it reduces to}\qquad\vcenter{\irule{\irulePi{\pi}{\Gamma \vdash C}{}}{\Gamma \vdash C}{\mbox{prod}}}$$
that is to $\pi$ with an added interstitial rule. The role of this
prod rule will be made clear when we add scalars.
       
The reduction rules of the connective $\odot$ are, as we have seen
$$\vcenter{\irule{\irule{\irulePi{\pi_1}{\Gamma \vdash A}{}
                &
                \irulePi{\pi_2}{\Gamma \vdash B}{}
               }
               {\Gamma \vdash A \odot B}
               {\mbox{$\odot$-i}}
          &
          \irulePi{\pi_3}
                {\Gamma, A \vdash C}{} & \irulePi{\pi_4}{\Gamma, B \vdash C}
                {}
        }
        {\Gamma \vdash C}
        {\mbox{$\odot$-e}}} \qquad
	\textrm{that reduces to}\qquad
        (\pi_1/A)\pi_3 \textrm{ and } (\pi_2/B)\pi_4$$
in a non-deterministic way, erasing some information
$$\vcenter{\irule{\irule{\irulePi{\pi_1}{\Gamma \vdash A}{}
                &
                \irulePi{\pi_2}{\Gamma \vdash B}{}
               }
               {\Gamma \vdash A \odot B}
               {\mbox{$\odot$-i}}
          &
          \irulePi{\pi_3}
                {\Gamma, A \vdash C}
                {}
        }
        {\Gamma \vdash C}
        {\mbox{$\odot$-e1}}}
\qquad \textrm{that reduces to} \qquad (\pi_1/A)\pi_3$$
and
$$\vcenter{\irule{\irule{\irulePi{\pi_1}{\Gamma \vdash A}{}
                &
                \irulePi{\pi_2}{\Gamma \vdash B}{}
               }
               {\Gamma \vdash A \odot B}
               {\mbox{$\odot$-i}}
          &
          \irulePi{\pi_3}
                {\Gamma, B \vdash C}
                {}
        }
        {\Gamma \vdash C}
        {\mbox{$\odot$-e2}}} \qquad
\textrm{that reduces to} \qquad (\pi_2/A)\pi_3$$

Finally, adding the interstitial rules, permits to build proofs that cannot be
reduced, because the sum rule or the prod rule creates an interstice
between the introduction rule of some connective and its
elimination rule.  For example
$$\irule{\irule{\irule{\irulePi{\pi_1}{\Gamma \vdash A}{}
                 & 
                \irulePi{\pi_2}{\Gamma \vdash B}{}}
               {\Gamma \vdash A \wedge B}
               {\mbox{$\wedge$-i}}
         &
         \irule{\irulePi{\pi_3}{\Gamma \vdash A}{}
                &
                \irulePi{\pi_4}{\Gamma \vdash B}{}
         }
         {\Gamma \vdash A \wedge B}
         {\mbox{$\wedge$-i}}
        }
        {\Gamma \vdash A \wedge B}
        {\mbox{sum}}
        &
        \irulePi{\pi_5}{\Gamma, A \vdash C}{}
        }
        {\Gamma \vdash C}
        {\mbox{$\wedge$-e1}}$$
Reducing such a proof requires rules to commute the rule sum either
with the elimination rule below or with the introduction rules above.

As the commutation with the introduction rules
above is not always possible, for example in the proof
$$\irule{\irule{\irulePi{\pi_1}{\Gamma \vdash A}{}}
               {\Gamma \vdash A \vee B}
               {\mbox{$\vee$-i1}}
         &
         \irule{\irulePi{\pi_2}{\Gamma \vdash B}{}}
         {\Gamma \vdash A \vee B}
         {\mbox{$\vee$-i2}}
        }
        {\Gamma \vdash A \vee B}
        {\mbox{sum}}$$
the commutation with the elimination rule below is often preferred.
In this paper, we favour the commutation of the interstitial rules with
the introduction rules, rather than with the elimination rules,
whenever it is possible, that is for all connectives except
disjunction. For example the proof
$$\irule{\irule{\irulePi{\pi_1}{\Gamma \vdash A}{}
                &
                \irulePi{\pi_2}{\Gamma \vdash B}{}}
               {\Gamma \vdash A \wedge B}
               {\mbox{$\wedge$-i}}
         &
         \irule{\irulePi{\pi_3}{\Gamma \vdash A}{}
                &
                \irulePi{\pi_4}{\Gamma \vdash B}{}
         }
         {\Gamma \vdash A \wedge B}
         {\mbox{$\wedge$-i}}
        }
        {\Gamma \vdash A \wedge B}
        {\mbox{sum}}$$
reduces to 
$$\irule{\irule{\irulePi{\pi_1}{\Gamma \vdash A}{}
                &
                \irulePi{\pi_3}{\Gamma \vdash A}{}}
               {\Gamma \vdash A}
               {\mbox{sum}}
         &
         \irule{\irulePi{\pi_2}{\Gamma \vdash B}{}
                &
                \irulePi{\pi_4}{\Gamma \vdash B}{}
         }
         {\Gamma \vdash B}
         {\mbox{sum}}
        }
        {\Gamma \vdash A \wedge B}
        {\mbox{$\wedge$-i}}$$
Such a commutation yields a stronger introduction property for the
considered connective (Theorem~\ref{introductions}).

In the proof 
$$\irule{\irule{\irulePi{\pi}{\Gamma \vdash A}{}}
               {\Gamma \vdash A \vee B}
               {\mbox{$\vee$-i1}}
        }
        {\Gamma \vdash A \vee B}
        {\mbox{prod}}$$
the prod rule and the $\vee$-i1 rule can be commuted.  For
coherence, we have decided to commute both the sum rule and the prod
rule with the elimination rule of the disjunction, rather that with
its introduction rules, but both choices are possible.

\subsection{Proof-terms}

We introduce a term language, the $\odot$-calculus, for the proofs of
this logic. Its syntax is 
\begin{align*}
  t =~& x \mid t \plus t \mid \bullet t \mid \star
  \mid \elimtop(t,t) \mid \elimbot(t)\\
        &\mid \lambda \abstr{x}t\mid t~t
  \mid \pair{t}{t} \mid \elimand^1(t,\abstr{x}t)
  \mid \elimand^2(t,\abstr{x}t)
  \\
        &\mid \inl(t)\mid \inr(t) \mid \elimor(t,\abstr{x}t,\abstr{x}t)\\
  &\mid \super{t}{t}\mid \elimsup(t,\abstr{x}t,\abstr{x}t)
  \mid \elimsup^1(t,\abstr{x}t)
  \mid \elimsup^2(t,\abstr{x}t)
\end{align*}
The variables $x$ express the proofs built with the rule axiom,
the terms $t \plus u$ those built with the rule sum,
the terms $\bullet t$ those built with the rule prod, 
the term $\star$ that built with the rule $\top$-i,
the terms $\elimtop(t,u)$ those built with the rule $\top$-e,
the terms $\elimbot(t)$ those built with the rule $\bot$-e,
the terms $\lambda \abstr{x}t$ those built with the rule $\Rightarrow$-i,
the terms $t~u$ those built with the rule $\Rightarrow$-e,
the terms $\pair{t}{u}$ those built with the rule $\wedge$-i,
the terms $\elimand^1(t, \abstr{x}u)$ and $\elimand^2(t, \abstr{x}u)$
those built with the rules $\wedge$-e1 and $\wedge$-e2,
the terms $\inl(t)$ 
and $\inr(t)$ those built with the rules $\vee$-i1 and $\vee$-i2,
the terms $\elimor(t,\abstr{x}u,\abstr{y}v)$ those built with the rule $\vee$-e,
the terms $\super{t}{u}$ those built with the rule $\odot$-i,
and the terms $\elimsup(t,\abstr{x}u,\abstr{y}v)$, 
$\elimsup^1(t,\abstr{x}u)$, and $\elimsup^2(t,\abstr{x}u)$
those built with the rules $\odot$-e, $\odot$-e1, and $\odot$-e2.

The proofs of the form $\star$, $\lambda \abstr{x}t$, $\pair{t}{u}$, $\inl(t)$,
$\inr(t)$, and $\super{t}{u}$ are called {\it introductions}, and those of
the form
$\elimtop(t,u)$,
$\elimbot(t)$,
$t~u$,
$\elimand^1(t,\abstr{x}u)$,
$\elimand^2(t,\abstr{x}u)$,
$\elimor(t,\abstr{x}u,\abstr{y}v)$,
$\elimsup(t,\abstr{x}u,\abstr{y}v)$, 
$\elimsup^1(t,\abstr{x}u)$, and
$\elimsup^2(t,\abstr{x}u)$
{\it eliminations}.  The variables and the proofs of the form $t \plus
u$ and $\bullet t$ are neither introductions nor eliminations.

The $\alpha$-equivalence relation and the free and bound variables of
a proof-term are defined as usual. Proof-terms are defined modulo
$\alpha$-equivalence.  A proof-term is closed if it contains no free
variables.  We write $(u/x)t$ for the substitution of $u$ for $x$ in
$t$.

\begin{figure}[t]
  $$\irule{}
        {\Gamma \vdash x:A}
        {\mbox{axiom~$x:A \in \Gamma$}}
	\qquad
\irule{\Gamma \vdash t:A & \Gamma \vdash u:A}
        {\Gamma \vdash t \plus u:A}
        {\mbox{sum}}
        \qquad
\irule{\Gamma \vdash t:A}
        {\Gamma \vdash \bullet t:A}
        {\mbox{prod}}$$
$$\irule{}
        {\Gamma \vdash \star:\top}
        {\mbox{$\top$-i}}
        \qquad
        \irule{\Gamma \vdash t:\top & \Gamma \vdash u:C}
              {\Gamma \vdash \elimtop(t,u):C}
        {\mbox{$\top$-e}}
        \qquad
\irule{\Gamma \vdash t:\bot}
        {\Gamma \vdash \elimbot(t):C}
        {\mbox{$\bot$-e}}
	$$
$$\irule{\Gamma, x:A \vdash t:B}
        {\Gamma \vdash \lambda \abstr{x}t:A \Rightarrow B}
        {\mbox{$\Rightarrow$-i}}
        \qquad
\irule{\Gamma \vdash t:A\Rightarrow B & \Gamma \vdash u:A}
        {\Gamma \vdash t~u:B}
        {\mbox{$\Rightarrow$-e}}$$
$$\irule{\Gamma \vdash t:A & \Gamma \vdash u:B}
        {\Gamma \vdash \pair{t}{u}:A \wedge B}
        {\mbox{$\wedge$-i}}$$
$$\irule{\Gamma \vdash t:A \wedge B & \Gamma, x:A \vdash u:C}
        {\Gamma \vdash \elimand^1(t,\abstr{x}u):C}
        {\mbox{$\wedge$-e1}}
        \qquad
        \irule{\Gamma \vdash t:A \wedge B & \Gamma, x:B \vdash u:C}
        {\Gamma \vdash \elimand^2(t,\abstr{x}u):C}
        {\mbox{$\wedge$-e2}}$$
$$\irule{\Gamma \vdash t:A}
        {\Gamma \vdash \inl(t):A \vee B}
        {\mbox{$\vee$-i1}}
        \qquad
\irule{\Gamma \vdash t:B}
        {\Gamma \vdash \inr(t):A \vee B}
        {\mbox{$\vee$-i2}}$$
$$  
\irule{\Gamma \vdash t:A \vee B & \Gamma, x:A \vdash u:C & \Gamma, y:B \vdash v:C}
        {\Gamma \vdash \elimor(t,\abstr{x}u,\abstr{y}v):C}
        {\mbox{$\vee$-e}}$$
$$\irule{\Gamma \vdash t:A & \Gamma \vdash u:B}
        {\Gamma \vdash \super{t}{u}:A \odot B}
        {\mbox{$\odot$-i}}$$
$$\irule{\Gamma \vdash t:A\odot B & \Gamma, x:A \vdash u:C & \Gamma, y:B \vdash v:C}
        {\Gamma \vdash \elimsup(t,\abstr{x}u,\abstr{y}v):C}
        {\mbox{$\odot$-e}}$$
$$\irule{\Gamma \vdash t:A\odot B & \Gamma, x:A \vdash u:C}
        {\Gamma \vdash \elimsup^1(t,\abstr{x}u):C}
        {\mbox{$\odot$-e1}}
        \qquad
  \irule{\Gamma \vdash t:A\odot B & \Gamma, x:B \vdash u:C}
        {\Gamma \vdash \elimsup^2(t,\abstr{x}u):C}
        {\mbox{$\odot$-e2}}$$
\caption{The typing rules of the $\odot$-calculus\label{figtypingrules}}
\end{figure}
        
\begin{figure}[t]
  \[
    \begin{array}{r@{\,}l}
      \elimtop(\star, t) & \longrightarrow \bullet t\\
      (\lambda \abstr{x}t)~u & \longrightarrow  (u/x)t\\
      \elimand^1(\pair{t}{u}, \abstr{x}v) & \longrightarrow  (t/x)v\\
      \elimand^2(\pair{t}{u}, \abstr{x}v) & \longrightarrow  (u/x)v\\
      \elimor(\inl(t),\abstr{x}v,\abstr{y}w) & \longrightarrow  (t/x)v\\
      \elimor(\inr(u),\abstr{x}v,\abstr{y}w) & \longrightarrow  (u/y)w\\
      \elimsup(\super{t}{u},\abstr{x}v,\abstr{y}w) & \longrightarrow  (t/x)v\\
      \elimsup(\super{t}{u},\abstr{x}v,\abstr{y}w) & \longrightarrow  (u/y)w\\
      \elimsup^1(\super{t}{u},\abstr{x}v) & \longrightarrow  (t/x)v\\
      \elimsup^2(\super{t}{u},\abstr{x}v) & \longrightarrow  (u/x)v\\
      \\
      \star \plus \star&\longrightarrow  \star\\
      (\lambda \abstr{x}t) \plus (\lambda \abstr{x}u) & \longrightarrow  \lambda \abstr{x}(t \plus u)\\
      \pair{t}{u} \plus \pair{v}{w}
      & \longrightarrow  \pair{t \plus v}{u \plus w}\\
      \elimor(t \plus u,\abstr{x}v,\abstr{y}w) & \longrightarrow 
      \elimor(t,\abstr{x}v,\abstr{y}w)
      \plus
      \elimor(u,\abstr{x}v,\abstr{y}w)
      \\
      \super{t}{u} \plus \super{v}{w} & \longrightarrow
      \super{t \plus v}{u \plus w}\\
\\
      \bullet \star &\longrightarrow  \star\\
      \bullet (\lambda \abstr{x}t) & \longrightarrow  \lambda \abstr{x}(\bullet t)\\
      \bullet \pair{t}{u} 
      & \longrightarrow  \pair{\bullet t}{\bullet u}\\
      \elimor(\bullet t,\abstr{x}v,\abstr{y}w) & \longrightarrow 
      \bullet \elimor(t,\abstr{x}v,\abstr{y}w)
      \\
      \bullet \super{t}{u} & \longrightarrow  \super{\bullet t}{\bullet u}
    \end{array}
  \]
\caption{The reduction rules of the $\odot$-calculus\label{figureductionrules}}
\end{figure}

The typing rules of the $\odot$-calculus are given in
Figure~\ref{figtypingrules} and its reduction rules in
Figure~\ref{figureductionrules}, they are the expression on proof-terms of the
reduction rules presented above. An instance of the second rule of
Figure~\ref{figureductionrules} is $$(\lambda \abstr{x}x)~y \lra y$$ and an
instance of the third is $$\elimand^1(\pair{\star}{\star},\abstr{x}x) \lra
\star$$

In this paper, we consider the usual reduction rules for natural
deduction, but not the so-called {\it commuting cuts}, that yield the
subformula property and the equivalence with cut-free sequent
calculus. Although we believe proving the termination with commuting cuts
is not difficult (in particular because we have included ultra-reduction
rules), we leave this for future work.

\begin{remark}
This system is a higher-order rewrite system \cite{Klop,Nipkow}. A more 
rigorous notation would be to consider 
the symbol
$\abstr{}$ as the abstraction, to add a symbol $app$ for application, and 
to add a rewrite rule $\beta$, $app(\abstr{x}t,u) \lra (u/x)t$ used to build 
the instances of the rules of Figure~\ref{figureductionrules}.

Hence, the second rule of Figure~\ref{figureductionrules} would be expressed as
$$(\lambda \abstr{x}app(T,x))~U \lra app(T,U)$$
the third as 
$$\elimand^1(\pair{T}{U},\abstr{x}app(V,x)) \lra app(V,T)$$
etc.

Substituting the proof $\abstr{x}x$ for the variable $T$ and the proof $y$ for
the variable $U$ in the first rule and reducing, with the added rule, both
sides of the rule yields the instance $$(\lambda \abstr{x}x)~y \lra y$$
and substituting $\star$ for the variables $T$ and $U$ and the proof
$\abstr{x}x$ for the variable $V$, in the second, and
reducing both sides of the rule, with the added rule, yields the
instance
$$\elimand^1(\pair{\star}{\star},\abstr{x}x) \lra \star$$
etc.
\end{remark}

\subsection{Subject reduction}

To prove subject reduction, we first prove, as usual, a substitution
property.

\begin{proposition}[Substitution]
  If $\Gamma,x:B\vdash t:A$ and $\Gamma\vdash u:B$, then $\Gamma \vdash
  (u/x)t:A$.
\end{proposition}

\begin{proof}
By induction on the structure of $t$.
\end{proof}

And use it to prove the theorem itself.

\begin{theorem}[Subject reduction]
  If $\Gamma \vdash t:A$ and $t \lra u$, then $\Gamma \vdash u:A$.
\end{theorem}

\begin{proof}
By induction on the definition of the relation $\lra$.
\end{proof}

\subsection{Confluence}

The introduction of the connective $\odot$ leads to a non-deterministic
calculus with the rules
$$\elimsup(\super{t}{u},\abstr{x}v,\abstr{y}w) \longrightarrow  (t/x)v$$
$$\elimsup(\super{t}{u},\abstr{x}v,\abstr{y}w) \longrightarrow
(u/y)w$$ Hence, the system presented in Figure~\ref{figureductionrules} is trivially non-confluent.  But if we drop
these two rules, the rest of the system is confluent.

\begin{theorem}[Confluence]
The system of Figure~\ref{figureductionrules} without the rules 
$$\elimsup(\super{t}{u},\abstr{x}v,\abstr{y}w) \longrightarrow  (t/x)v$$
$$\elimsup(\super{t}{u},\abstr{x}v,\abstr{y}w) \longrightarrow  (u/y)w$$
is confluent.
\end{theorem}

\begin{proof}
This system is left linear and it has no
critical pairs. And, as proved in \cite[Theorem 6.8]{Nipkow}, 
higher-order left linear systems without
critical pairs are confluent.
\end{proof}

Note that the untyped calculus does have critical pairs, for example the
proof $\elimor(\star \plus \star,\abstr{x}x,\abstr{y}y)$ reduces in
two different ways, but these critical pairs are not well-typed.

\subsection{Termination}

We now prove the strong termination of proof reduction, that is that
all reduction sequences are finite.  The proof follows the same
pattern as that for
intuitionistic
propositional natural deduction, that we recall in
the Appendix.

The $\odot$-calculus introduces two new features: the connective
$\odot$, its associated proof constructors $\super{}{}$, $\elimsup$,
$\elimsup^1$, and $\elimsup^2$, and the constructors $\plus$ and $\bullet$.
The termination proof of intuitionistic
propositional natural deduction extends smoothly
when we add the connective $\odot$, but adding the constructors $\plus$
and $\bullet$
is a bit more challenging.  To handle these symbols, we prove the strong
termination of an extended reduction system, in the spirit of Girard's
ultra-reduction \cite{Girard}, whose strong termination obviously
implies that of the rules of Figure~\ref{figureductionrules}.

\begin{definition}[Ultra-reduction]
  Ultra-reduction is defined with the rules of Figure~\ref{figureductionrules},
  plus the rules
  \begin{align*}
    t \plus u & \longrightarrow  t\\
    t \plus u & \longrightarrow  u\\
    \bullet t & \longrightarrow  t
  \end{align*}
\end{definition}

In the proof below, Propositions~\ref{Var}, \ref{star},
\ref{abstraction}, \ref{pair}, \ref{inl}, \ref{inr}, \ref{elimbot},
\ref{application}, \ref{elimand1}, and~\ref{elimand2} have the same
proofs as Propositions~\ref{app-Var}, \ref{app-star},
\ref{app-abstraction}, \ref{app-pair}, \ref{app-inl}, \ref{app-inr},
\ref{app-elimbot}, \ref{app-application}, \ref{app-elimand1}, and
\ref{app-elimand2} in the strong termination of proof reduction for
intuitionistic propositional natural deduction (except that the references to
Propositions~\ref{app-Var}, \ref{app-closure}, and~\ref{app-CR3} must
be replaced with references to Propositions~\ref{Var}, \ref{closure},
and~\ref{CR3}). So we will omit these proofs. Propositions
\ref{closure}, \ref{CR3}, \ref{elimtop}, and~\ref{elimor} have proofs
similar to those of Propositions~\ref{app-closure}, \ref{app-CR3},
\ref{app-elimtop}, and~\ref{app-elimor}, but these proofs require
minor tweaks. In contrast, Propositions~\ref{terminationsum},
\ref{terminationprod},
\ref{sum},
\ref{prod},
\ref{super},
\ref{elimsup}, \ref{elimsup1}, and~\ref{elimsup2} are specific.

\begin{definition}[Length of reduction]
If $t$ is a strongly terminating proof, we write $|t|$ for the
maximum length of a reduction sequence issued from $t$.
\end{definition}

\begin{proposition}[Termination of a sum]
\label{terminationsum}
If $t$ and $u$ strongly terminate, then so does $t \plus u$. 
\end{proposition}

\begin{proof}
We prove that all the one-step reducts of 
$t \plus u$ strongly terminate, by induction first on 
$|t| + |u|$ and then on the size of $t$. 

If the reduction takes place in $t$ or in $u$ we apply the induction
hypothesis.
Otherwise the reduction is at the root and the rule used is either
\begin{align*}
\star \plus \star &\longrightarrow  \star\\
(\lambda \abstr{x}t') \plus (\lambda \abstr{x}u')
&\longrightarrow  \lambda \abstr{x}(t' \plus u')\\
\pair{t'_1}{t'_2} \plus \pair{u'_1}{u'_2}
   &\longrightarrow  \pair{t'_1 \plus u'_1}{t'_2 \plus u'_2}\\
\super{t'_1}{t'_2} \plus \super{u'_1}{u'_2}
&\longrightarrow  \super{t'_1 \plus u'_1}{t'_2 \plus u'_2}\\
t \plus u &\longrightarrow t\\
t \plus u &\longrightarrow u
\end{align*}
In the first case, the proof $\star$ is irreducible, hence it strongly terminates. In the second, 
by induction hypothesis, the proof $t' \plus u'$
strongly terminates, thus so does the proof
$\lambda \abstr{x}(t' \plus u')$.
In the third and the fourth, 
by induction hypothesis, the proofs
$t'_1 \plus u'_1$ and $t'_2 \plus u'_2$
strongly terminate, hence so do the proofs
$\pair{t'_1 \plus u'_1}{t'_2 \plus u'_2}$
and 
$\super{t'_1 \plus u'_1}{t'_2 \plus u'_2}$.
In the fifth and the sixth, the proofs $t$ and $u$ strongly terminate. 
\end{proof}

\begin{proposition}[Termination of a product]
\label{terminationprod}
If $t$ strongly terminates, then so does $\bullet t$. 
\end{proposition}

\begin{proof}
We prove that all the one-step reducts of 
$\bullet t$ strongly terminate, by induction first on 
$|t|$ and then on the size of $t$. 

If the reduction takes place in $t$, we apply the induction
hypothesis.
Otherwise the reduction is at the root and the rule used is either
\begin{align*}
\bullet \star &\longrightarrow  \star\\
\bullet (\lambda \abstr{x}t') 
&\longrightarrow  \lambda \abstr{x}\bullet t'\\
\bullet \pair{t'_1}{t'_2} 
   &\longrightarrow  \pair{\bullet t'_1}{\bullet t'_2}\\
\bullet \super{t'_1}{t'_2} 
&\longrightarrow  \super{\bullet t'_1}{\bullet t'_2}\\
\bullet t &\longrightarrow t
\end{align*}
In the first case, the proof $\star$ is irreducible, hence it strongly terminates. In the second, 
by induction hypothesis, the proof $\bullet t'$
strongly terminates, thus so does the proof
$\lambda \abstr{x}\bullet t'$.
In the third and the fourth, 
by induction hypothesis, the proofs
$\bullet t'_1$ and $\bullet t'_2$
strongly terminate, hence so do the proofs
$\pair{\bullet t'_1}{\bullet t'_2}$
and 
$\super{\bullet t'_1}{\bullet t'_2}$.
In the fifth, the proof $t$ strongly terminates. 
\end{proof}

\begin{definition}
We define, by induction on the proposition $A$, a set of proofs
$\llbracket A \rrbracket$:
\begin{itemize}
\item $t \in \llbracket \top \rrbracket$ if $t$ strongly terminates,

\item $t \in \llbracket \bot \rrbracket$ if $t$ strongly terminates,

\item $t \in \llbracket A \Rightarrow B \rrbracket$ if $t$ strongly
  terminates and whenever it reduces to a proof of the form $\lambda
  \abstr{x}u$, then for every $v \in \llbracket A \rrbracket$, $(v/x)u \in
  \llbracket B \rrbracket$,

\item $t \in \llbracket A \wedge B \rrbracket$ if $t$ strongly
  terminates and whenever it reduces to a proof of the form
  $\pair{u}{v}$, then $u \in \llbracket A \rrbracket$ and $v \in \llbracket
  B \rrbracket$,

\item $t \in \llbracket A \vee B \rrbracket$ if $t$ strongly
  terminates and whenever it reduces to a proof of the form $\inl(u)$,
  then $u \in \llbracket A \rrbracket$, and whenever it reduces to a
  proof of the form $\inr(v)$, then $v \in \llbracket B \rrbracket$,

\item $t \in \llbracket A \odot B \rrbracket$ if $t$ strongly
  terminates and whenever it reduces to a proof of the form $\super{u}{v}$,
  then $u \in \llbracket A \rrbracket$ and $v \in \llbracket B
  \rrbracket$.
\end{itemize}
\end{definition}

\begin{proposition}[Variables]
\label{Var}
For any $A$, the set $\llbracket A \rrbracket$ contains all the variables.
\end{proposition}

\begin{proposition}[Closure by reduction]
\label{closure}
If $t \in \llbracket A \rrbracket$ and $t \longrightarrow^* t'$, then 
$t' \in \llbracket A \rrbracket$.
\end{proposition}

\begin{proof}
If $t \longrightarrow^* t'$ and $t$ strongly terminates, then $t'$
strongly terminates.

Furthermore, if $A$ has the form $B \Rightarrow C$ and $t'$ reduces to
$\lambda \abstr{x}u$, then so does $t$, hence for every $v \in \llbracket B
\rrbracket$, $(v/x)u \in \llbracket C \rrbracket$.

If $A$ has the form $B \wedge C$ and $t'$ reduces to $\pair{u}{v}$,
then so does $t$, hence $u \in \llbracket B \rrbracket$ and $v \in
\llbracket C \rrbracket$.

If $A$ has the form $B \vee C$ and $t'$ reduces to $\inl(u)$, then so
does $t$, hence $u \in \llbracket B \rrbracket$ and if $A$ has the
form $B \vee C$ and $t'$ reduces to $\inr(v)$, then so does $t$, hence
$v \in \llbracket C \rrbracket$.

And if $A$ has the form $B \odot C$ and $t'$ reduces to $\super{u}{v}$, then
so does $t$, hence $u \in \llbracket B \rrbracket$ and $v \in
\llbracket C \rrbracket$.
\end{proof}

\begin{proposition}[Girard's lemma]
\label{CR3}
Let $t$ be a proof that is not an introduction, 
such that all the one-step reducts of $t$
are in $\llbracket A \rrbracket$. Then, $t \in \llbracket A \rrbracket$.
\end{proposition}

\begin{proof}
Let $t, t_2, \ldots$ be a reduction sequence issued from $t$. If it has a
single element, it is finite. Otherwise, we have $t \longrightarrow
t_2$. As $t_2 \in \llbracket A \rrbracket$, it strongly terminates and
the reduction sequence is finite. Thus, $t$ strongly terminates.

Furthermore, if $A$ has the form $B \Rightarrow C$ and $t
\longrightarrow^* \lambda \abstr{x}u$, then let $t , t_2, \ldots, t_n$ be a
reduction sequence from $t$ to $\lambda \abstr{x}u$.  As $t_n$ is an
introduction and $t$ is not, $n \geq 2$. Thus, $t \longrightarrow t_2
\longrightarrow^* t_n$. We have $t_2 \in \llbracket A \rrbracket$,
thus for all $v \in \llbracket B \rrbracket$, $(v/x)u \in \llbracket C
\rrbracket$.

If $A$ has the form $B \wedge C$ and $t \longrightarrow^* \pair{u}{v}$, then let $t , t_2, \ldots, t_n$ be a reduction sequence
from $t$ to $\pair{u}{v}$.  As $t_n$ is an introduction and
$t$ is not, $n \geq 2$. Thus, $t \longrightarrow t_2 \longrightarrow^*
t_n$. We have $t_2 \in \llbracket A \rrbracket$, thus $u \in
\llbracket B \rrbracket$ and $v \in \llbracket C \rrbracket$.

If $A$ has the form $B \vee C$ and $t \longrightarrow^* \inl(u)$, then let $t ,
t_2, \ldots, t_n$ be a reduction sequence from $t$ to $\inl(u)$.  As
$t_n$ is an introduction and $t$ is not, $n \geq 2$.  Thus, $t
\longrightarrow t_2 \longrightarrow^* t_n$. We have $t_2 \in
\llbracket A \rrbracket$, thus $u \in \llbracket B \rrbracket$.

If $A$ has the form $B \vee C$ and $t \longrightarrow^* \inr(v)$, then let $t ,
t_2, \ldots, t_n$ be a reduction sequence from $t$ to $\inr(v)$.  As
$t_n$ is an introduction and $t$ is not, $n \geq 2$.  Thus, $t
\longrightarrow t_2 \longrightarrow^* t_n$. We have $t_2 \in
\llbracket A \rrbracket$, thus $v \in \llbracket C \rrbracket$.

And if $A$ has the form $B \odot C$ and $t \longrightarrow^* \super{u}{v}$, then
let $t , t_2, \ldots, t_n$ be a reduction sequence from $t$ to $\super{u}{v}$.
As $t_n$ is an introduction and $t$ is not, $n \geq 2$. Thus, $t
\longrightarrow t_2 \longrightarrow^* t_n$. We have $t_2 \in
\llbracket A \rrbracket$, thus $u \in \llbracket B \rrbracket$ and $v
\in \llbracket C \rrbracket$.
\end{proof}

In Propositions~\ref{sum} to~\ref{elimsup2}, we prove the adequacy of each proof constructor.

\begin{proposition}[Adequacy of $\plus$]
\label{sum}
If $t_1 \in \llbracket A \rrbracket$ and $t_2 \in \llbracket A
\rrbracket$, then $t_1 \plus t_2 \in \llbracket A \rrbracket$.
\end{proposition}

\begin{proof}
  By induction on $A$.

  The proofs $t_1$ and $t_2$ strongly terminate.  Thus, by
  Proposition~\ref{terminationsum}, the proof $t_1 \plus t_2$ strongly
  terminates.

  Furthermore:
  \begin{itemize}
    \item If the proposition $A$ has the form $B \Rightarrow C$, and $t_1
      \plus t_2 \lra^* \lambda \abstr{x} v$ then either $t_1 \lra^*
      \lambda \abstr{x} u_1$, $t_2 \lra^* \lambda \abstr{x} u_2$, and $u_1
      \plus u_2 \lra^* v$, or $t_1 \lra^* \lambda \abstr{x} v$, or $t_2
      \lra^* \lambda \abstr{x} v$.
  
  In the first case, by Proposition~\ref{closure}, the proofs $\lambda
  \abstr{x}u_1$ and $\lambda \abstr{x}u_2$ are in $\llbracket A
  \rrbracket$. Thus, for every $w$ in $\llbracket B \rrbracket$,
  $(w/x)u_1 \in \llbracket C \rrbracket$ and $(w/x)u_2 \in \llbracket
  C \rrbracket$.  By induction hypothesis, $(w/x)(u_1 \plus u_2) =
  (w/x)u_1 \plus (w/x)u_2 \in \llbracket C \rrbracket$ and by
  Proposition~\ref{closure}, $(w/x)v \in \llbracket C \rrbracket$.

  In the second and the third, by Proposition~\ref{closure}, $\lambda
  \abstr{x} v \in \llbracket A \rrbracket$ hence, for every $w$ in
  $\llbracket B \rrbracket$, $(w/x)v \in \llbracket C \rrbracket$.
  
\item If the proposition $A$ has the form $B \wedge C$, and $t_1
  \plus t_2 \lra^* \pair{v}{v'}$ then $t_1 \lra^* \pair{u_1}{u'_1}$,
  $t_2 \lra^* \pair{u_2}{u'_2}$, $u_1 \plus u_2 \lra^* v$, and $u'_1
  \plus u'_2 \lra^* v'$, or $t_1 \lra^* \pair{v}{v'}$, or $t_2 \lra^*
  \pair{v}{v'}$.

  In the first case, by Proposition~\ref{closure}, the proofs
  $\pair{u_1}{u'_1}$ and $\pair{u_2}{u'_2}$ are in $\llbracket A
  \rrbracket$.  Thus, $u_1$ and $u_2$ are in $\llbracket B \rrbracket$
  and $u'_1$ and $u'_2$ are in $\llbracket C \rrbracket$.  By
  induction hypothesis, $u_1 \plus u_2 \in \llbracket B \rrbracket$
  and $u'_1 \plus u'_2 \in \llbracket C \rrbracket$ and by
  Proposition~\ref{closure}, $v \in \llbracket B \rrbracket$ and $v' \in
  \llbracket C \rrbracket$.

  In the second and the third, by Proposition~\ref{closure},
  $\pair{v}{v'} \in \llbracket A \rrbracket$ hence, $v \in \llbracket
  B \rrbracket$ and $v' \in \llbracket C \rrbracket$.

\item If the proposition $A$ has the form $B \vee C$, and $t_1 \plus
  t_2 \lra^* \inl(v)$ then $t_1 \lra^* \inl(v)$ or $t_2 \lra^*
  \inl(v)$.  Thus, by Proposition~\ref{closure}, $\inl(v) \in
  \llbracket A \rrbracket$, hence $v \in \llbracket B \rrbracket$.

  The proof is similar if $t_1 \plus t_2 \lra^* \inr(v)$.
  
\item If the proposition $A$ has the form $B \odot C$, and $t_1
  \plus t_2 \lra^* \super{v}{v'}$ then $t_1 \lra^* \super{u_1}{u'_1}$,
  $t_2 \lra^* \super{u_2}{u'_2}$, $u_1 \plus u_2 \lra^* v$, and $u'_1
  \plus u'_2 \lra^* v'$, or $t_1 \lra^* \super{v}{v'}$, or $t_2 \lra^*
  \super{v}{v'}$.

  In the first case, by Proposition~\ref{closure}, the proofs
  $\super{u_1}{u'_1}$ and $\super{u_2}{u'_2}$ are in $\llbracket A
  \rrbracket$.  Thus, $u_1$ and $u_2$ are in $\llbracket B \rrbracket$
  and $u'_1$ and $u'_2$ are in $\llbracket C \rrbracket$.  By
  induction hypothesis, $u_1 \plus u_2 \in \llbracket B \rrbracket$
  and $u'_1 \plus u'_2 \in \llbracket C \rrbracket$ and by
  Proposition~\ref{closure}, $v \in \llbracket B \rrbracket$ and $v' \in
  \llbracket C \rrbracket$.

  In the second and the third, by Proposition~\ref{closure},
  $\super{v}{v'} \in \llbracket A \rrbracket$ hence, $v \in \llbracket
  B \rrbracket$ and $v' \in \llbracket C \rrbracket$.  \qedhere
\end{itemize}
\end{proof}

\begin{proposition}[Adequacy of $\bullet$]
\label{prod}
If $t \in \llbracket A \rrbracket$, then $\bullet t \in \llbracket
A \rrbracket$.
\end{proposition}

\begin{proof}
  By induction on $A$.

  The proofs $t$ strongly terminates.  Thus, by
  Proposition~\ref{terminationprod}, the proof $\bullet t$ strongly terminates.

  Furthermore:

\begin{itemize}
\item If the proposition $A$ has the form $B \Rightarrow C$, and $\bullet t
  \lra^* \lambda \abstr{x} v$ then either $t \lra^*
  \lambda \abstr{x} u$, and $\bullet u
   \lra^* v$, or $t \lra^* \lambda \abstr{x} v$.
  
  In the first case, by Proposition~\ref{closure}, the proof $\lambda
  \abstr{x}u$ is in $\llbracket A \rrbracket$. Thus, for every $w$ in
  $\llbracket B \rrbracket$, $(w/x)u \in \llbracket C \rrbracket$.
  By induction hypothesis, $(w/x)\bullet u = \bullet (w/x)u \in
  \llbracket C \rrbracket$ and by Proposition~\ref{closure}, $(w/x)v
  \in \llbracket C \rrbracket$.

  In the second, by Proposition~\ref{closure}, $\lambda
  \abstr{x} v \in \llbracket A \rrbracket$ hence, for every $w$ in
  $\llbracket B \rrbracket$, $(w/x)v \in \llbracket C \rrbracket$.
  
\item If the proposition $A$ has the form $B \wedge C$, and $\bullet t
  \lra^* \pair{v}{v'}$ then $t \lra^* \pair{u}{u'}$,
  $\bullet u \lra^* v$, and $\bullet u'
  \lra^* v'$, or $t \lra^* \pair{v}{v'}$.

  In the first case, by Proposition~\ref{closure}, the proof
  $\pair{u}{u'}$ is in $\llbracket A
  \rrbracket$.  Thus, $u$ is in $\llbracket B \rrbracket$
  and $u'$ is in $\llbracket C \rrbracket$.  By
  induction hypothesis, $\bullet u \in \llbracket B \rrbracket$
  and $\bullet u' \in \llbracket C \rrbracket$ and by
  Proposition~\ref{closure}, $v \in \llbracket B \rrbracket$ and $v' \in
  \llbracket C \rrbracket$.

  In the second, by Proposition~\ref{closure},
  $\pair{v}{v'} \in \llbracket A \rrbracket$ hence, $v \in \llbracket
  B \rrbracket$ and $v' \in \llbracket C \rrbracket$.

\item
If the proposition $A$ has the form $B \vee C$, and $\bullet t
 \lra^* \inl(v)$ then $t \lra^* \inl(v)$.
  Thus, by Proposition~\ref{closure},
  $\inl(v) \in \llbracket A \rrbracket$, hence $v \in \llbracket
  B \rrbracket$.

  The proof is similar if $\bullet t \lra^* \inr(v)$.
  
\item If the proposition $A$ has the form $B \odot C$, and $\bullet t
   \lra^* \super{v}{v'}$ then $t \lra^* \super{u}{u'}$,
  $\bullet u  \lra^* v$, and $\bullet u'
  \lra^* v'$, or $t \lra^* \super{v}{v'}$.

  In the first case, by Proposition~\ref{closure}, the proofs
  $\super{u}{u'}$ is in $\llbracket A
  \rrbracket$.  Thus, $u$ is in $\llbracket B \rrbracket$
  and $u'$ is in $\llbracket C \rrbracket$.  By
  induction hypothesis, $\bullet u  \in \llbracket B \rrbracket$
  and $\bullet u' \in \llbracket C \rrbracket$ and by
  Proposition~\ref{closure}, $v \in \llbracket B \rrbracket$ and $v' \in
  \llbracket C \rrbracket$.

  In the second, by Proposition~\ref{closure},
  $\super{v}{v'} \in \llbracket A \rrbracket$ hence, $v \in \llbracket
  B \rrbracket$ and $v' \in \llbracket C \rrbracket$.  \qedhere
\end{itemize}
\end{proof}

\begin{proposition}[Adequacy of $\star$]
\label{star}
We have $\star \in \llbracket \top \rrbracket$.
\end{proposition}

\begin{proposition}[Adequacy of $\lambda$]
\label{abstraction}
If, for all $u \in \llbracket A \rrbracket$, $(u/x)t \in \llbracket B
\rrbracket$, then $\lambda \abstr{x}t \in \llbracket A \Rightarrow B
\rrbracket$.
\end{proposition}

\begin{proposition}[Adequacy of $\pair{}{}$]
\label{pair}
If $t_1 \in \llbracket A \rrbracket$ and $t_2 \in \llbracket B
\rrbracket$, then $\pair{t_1}{t_2} \in \llbracket A \wedge B
\rrbracket$.
\end{proposition}

\begin{proposition}[Adequacy of $\inl$]
\label{inl}
If $t \in \llbracket A \rrbracket$, then $\inl(t) \in \llbracket A
\vee B \rrbracket$.
\end{proposition}

\begin{proposition}[Adequacy of $\inr$]
\label{inr}
If $t \in \llbracket B \rrbracket$, then $\inr(t) \in \llbracket A
\vee B \rrbracket$.
\end{proposition}

\begin{proposition}[Adequacy of $\super{}{}$]
\label{super}
If $t_1 \in \llbracket A \rrbracket$ and $t_2 \in \llbracket B
\rrbracket$, then $\super{t_1}{t_2} \in \llbracket A \odot B \rrbracket$.
\end{proposition}

\begin{proof}
The proofs $t_1$ and $t_2$ strongly terminate. Consider a reduction
sequence issued from $\super{t_1}{t_2}$.  This sequence can only reduce $t_1$
and $t_2$, hence it is finite.  Thus, $\super{t_1}{t_2}$ strongly terminates.

Furthermore, if $\super{t_1}{t_2} \longrightarrow^*
\super{t'_1}{t'_2}$, then $t_1 \lra^* t'_1$ and $t_2 \lra^* t'_2$.  By
Proposition~\ref{closure}, $t'_1 \in \llbracket A \rrbracket$ and
$t'_2 \in \llbracket B \rrbracket$.
\end{proof}

\begin{proposition}[Adequacy of $\elimtop$]
\label{elimtop}
If $t_1 \in \llbracket \top \rrbracket$ and $t_2 \in \llbracket C \rrbracket$, 
then $\elimtop(t_1,t_2) \in \llbracket C \rrbracket$.
\end{proposition}

\begin{proof}
The proofs $t_1$ and $t_2$ strongly terminate.  We prove, by
induction on $|t_1| + |t_2|$, that $\elimtop(t_1,t_2)
\in \llbracket C \rrbracket$.  Using Proposition~\ref{CR3}, we only
need to prove that every of its one step reducts is in $\llbracket C
\rrbracket$.  If the reduction takes place in $t_1$ or $t_2$, then we
apply Proposition~\ref{closure} and the induction hypothesis.

Otherwise, the proof $t_1$ is $\star$ and the
reduct is $\bullet t_2$. We conclude with Proposition~\ref{prod}.
\end{proof}

\begin{proposition}[Adequacy of $\elimbot$]
\label{elimbot}
If $t \in \llbracket \bot \rrbracket$, 
then $\elimbot(t) \in \llbracket C \rrbracket$.
\end{proposition}

\begin{proposition}[Adequacy of application]
\label{application}
If $t_1 \in \llbracket A \Rightarrow B \rrbracket$ and $t_2 \in
\llbracket A \rrbracket$, then $t_1~t_2 \in \llbracket B
\rrbracket$.
\end{proposition}

\begin{proposition}[Adequacy of $\elimand^1$]
\label{elimand1}
If $t_1 \in \llbracket A \wedge B \rrbracket$
and, for all $u$ in $\llbracket A \rrbracket$,
$(u/x)t_2 \in \llbracket C \rrbracket $, 
then $\elimand^1(t_1, \abstr{x}t_2) \in \llbracket C \rrbracket$.
\end{proposition}

\begin{proposition}[Adequacy of $\elimand^2$]
\label{elimand2}
If $t_1 \in \llbracket A \wedge B \rrbracket$ and,
for all $u$ in $\llbracket B \rrbracket$,
$(u/x)t_2 \in \llbracket C \rrbracket $, 
then $\elimand^2(t_1, \abstr{x}t_2) \in \llbracket C \rrbracket$.
\end{proposition}

\begin{proposition}[Adequacy of $\elimor$]
\label{elimor}
If $t_1 \in \llbracket A \vee B \rrbracket$, for all $u$ in $\llbracket A
\rrbracket$, $(u/x)t_2 \in \llbracket C \rrbracket $, and, for all $v$
in $\llbracket B \rrbracket$, $(v/y)t_3 \in \llbracket C \rrbracket $,
then $\elimor(t_1, \abstr{x}t_2, \abstr{y}t_3) \in \llbracket C \rrbracket$.
\end{proposition}

\begin{proof}
By Proposition~\ref{Var}, $x \in \llbracket A \rrbracket$, thus $t_2 =
(x/x)t_2 \in \llbracket C \rrbracket$. In the same way, $t_3 \in
\llbracket C \rrbracket$.  Hence, $t_1$, $t_2$, and $t_3$ strongly
terminate.  We prove, by induction on $|t_1| + |t_2| + |t_3|$,
that $\elimor(t_1, \abstr{x}t_2,
\abstr{y}t_3) \in \llbracket C \rrbracket$.  Using Proposition~\ref{CR3}, we
only need to prove that every of its one step reducts
is in $\llbracket C \rrbracket$.  If the reduction takes place in
$t_1$, $t_2$, or $t_3$, then we apply Proposition~\ref{closure} and
the induction hypothesis. Otherwise, either:
\begin{itemize}
\item The proof $t_1$ has the form $\inl(w_2)$ and the reduct is
  $(w_2/x)t_2$. As $\inl(w_2) \in \llbracket A \vee B \rrbracket$, we
  have $w_2 \in \llbracket A \rrbracket$.  Hence, $(w_2/x)t_2 \in
  \llbracket C \rrbracket$.

\item The proof $t_1$ has the form $\inr(w_3)$ and the reduct is
  $(w_3/x)t_3$. As $\inr(w_3) \in \llbracket A \vee B \rrbracket$, we
  have $w_3 \in \llbracket B \rrbracket$.  Hence, $(w_3/x)t_3 \in
  \llbracket C \rrbracket$.

\item The proof $t_1$ has the form $t_1' \plus t''_1$ and the
  reduct is $\elimor(t'_1, \abstr{x}t_2, \abstr{y}t_3) \plus
  \elimor(t''_1, \abstr{x}t_2, \abstr{y}t_3)$. As $t_1
  \longrightarrow t'_1$ with an ultra-reduction rule, we have by
  Proposition~\ref{closure}, $t'_1 \in \llbracket A \vee B
  \rrbracket$.  In a similar way, $t''_1 \in \llbracket A \vee B
  \rrbracket$.  Thus, by induction hypothesis, $\elimor(t'_1,
  \abstr{x}t_2, \abstr{y}t_3) \in \llbracket A \vee B \rrbracket$
  and $\elimor(t''_1, \abstr{x}t_2, \abstr{y}t_3) \in \llbracket A
  \vee B \rrbracket$.  We conclude with Proposition~\ref{sum}.

\item The proof $t_1$ has the form $\bullet t_1'$ and the
  reduct is $\bullet \elimor(t'_1, \abstr{x}t_2, \abstr{y}t_3)$. As $t_1
  \longrightarrow t'_1$ with an ultra-reduction rule, we have by
  Proposition~\ref{closure}, $t'_1 \in \llbracket A \vee B
  \rrbracket$.  
  Thus, by induction hypothesis, $\elimor(t'_1,
  \abstr{x}t_2, \abstr{y}t_3) \in \llbracket A \vee B \rrbracket$.
  We conclude with Proposition~\ref{prod}.
  \qedhere
\end{itemize}
\end{proof}

\begin{proposition}[Adequacy of $\elimsup$]
\label{elimsup}
If $t_1 \in \llbracket A \odot B \rrbracket$, for all $u$ in $\llbracket A
\rrbracket$, $(u/x)t_2 \in \llbracket C \rrbracket $, and, for all $v$
in $\llbracket B \rrbracket$, $(v/y)t_3 \in \llbracket C \rrbracket $,
then $\elimsup(t_1, \abstr{x}t_2, \abstr{y}t_3) \in \llbracket C \rrbracket$.
\end{proposition}

\begin{proof}
By Proposition~\ref{Var}, $x \in \llbracket A \rrbracket$, thus $t_2 =
(x/x)t_2 \in \llbracket C \rrbracket$. In the same way, $t_3 \in
\llbracket C \rrbracket$.  Hence, $t_1$, $t_2$, and $t_3$ strongly
terminate.  We prove, by induction on $|t_1| + |t_2| + |t_3|$, that
$\elimsup(t_1, \abstr{x}t_2, \abstr{y}t_3) \in \llbracket C
\rrbracket$.  Using Proposition~\ref{CR3}, we only need to prove that
every of its one step reducts is in $\llbracket C \rrbracket$.  If
the reduction takes place in $t_1$, $t_2$, or $t_3$, then we apply
Proposition~\ref{closure} and the induction hypothesis.

Otherwise, the proof $t_1$ has the form $\super{u}{v}$ and the reduct is
either $(u/x)t_2$ or $(v/x)t_3$.  As $\super{u}{v} \in \llbracket A
\odot B \rrbracket$, we have $u \in \llbracket A \rrbracket$ and $v
\in \llbracket B \rrbracket$.  Hence, $(u/x)t_2 \in \llbracket C
\rrbracket$ and $(v/x)t_3 \in \llbracket C \rrbracket$.
\end{proof}

\begin{proposition}[Adequacy of $\elimsup^1$]
\label{elimsup1}
If $t_1 \in \llbracket A \odot B \rrbracket$ and, for all $u$ in
$\llbracket A \rrbracket$, $(u/x)t_2 \in \llbracket C \rrbracket $,
then $\elimsup^1(t_1, \abstr{x}t_2) \in \llbracket C \rrbracket$.
\end{proposition}

\begin{proof}
By Proposition~\ref{Var}, $x \in \llbracket A \rrbracket$, thus $t_2 =
(x/x)t_2 \in \llbracket C \rrbracket$. 
Hence, $t_1$ and $t_2$ strongly
terminate.  We prove, by induction on $|t_1| + |t_2|$, that
$\elimsup^1(t_1, \abstr{x}t_2) \in \llbracket C
\rrbracket$.  Using Proposition~\ref{CR3}, we only need to prove that
every of its one step reducts is in $\llbracket C \rrbracket$.  If
the reduction takes place in $t_1$ or $t_2$, then we apply
Proposition~\ref{closure} and the induction hypothesis.

Otherwise, the proof $t_1$ has the form $\super{u}{v}$ and the reduct is
$(u/x)t_2$.  As $\super{u}{v} \in \llbracket A
\odot B \rrbracket$, we have $u \in \llbracket A \rrbracket$.
Hence, $(u/x)t_2 \in \llbracket
C \rrbracket$.
\end{proof}

\begin{proposition}[Adequacy of $\elimsup^2$]
\label{elimsup2}
If $t_1 \in \llbracket A \odot B \rrbracket$ and, for all $u$ in
$\llbracket B \rrbracket$, $(u/x)t_2 \in \llbracket C \rrbracket $,
then $\elimsup^2(t_1, \abstr{x}t_2) \in \llbracket C \rrbracket$.
\end{proposition}

\begin{proof}
  Similar to that of Proposition~\ref{elimsup1}.

\end{proof}

\begin{theorem}[Adequacy]
Let $t$ be a proof of $A$ in a context $\Gamma = x_1:A_1, \ldots, x_n:A_n$ and
$\sigma$ be a substitution mapping each variable $x_i$ to an element
of $\llbracket A_i \rrbracket$, then $\sigma t \in \llbracket A
\rrbracket$.
\end{theorem}

\begin{proof}
By induction on the structure of $t$.

If $t$ is a variable, then, by definition of $\sigma$, $\sigma t \in
\llbracket A \rrbracket$.  For the seventeen other proof constructors,
we use the Propositions~\ref{star} to~\ref{elimsup2}.  As all cases
are similar, we just give a few examples.

\begin{itemize}
\item If $t = \super{u}{v}$, where $u$ is a proof of $B$ and $v$ a proof of
  $C$, then, by induction hypothesis, $\sigma u \in \llbracket B
  \rrbracket$ and $\sigma v \in \llbracket C \rrbracket$.  Hence, by
  Proposition~\ref{super}, $\super{\sigma u}{\sigma v} \in \llbracket B \odot C
  \rrbracket$, that is $\sigma t \in \llbracket A \rrbracket$.

\item If $t = \elimsup(u_1,\abstr{x}u_2,\abstr{y}u_3)$, where
  $u_1$ is a proof of $B \odot C$, $u_2$ a proof of $A$, and $u_3$ 
  a proof of $A$, then, by induction hypothesis, $\sigma u_1 \in
  \llbracket B \odot C \rrbracket$, for all $v$ in $\llbracket B
  \rrbracket$, $(v/x)\sigma u_2 \in \llbracket A \rrbracket$, and for
  all $w$ in $\llbracket C \rrbracket$, $(w/x)\sigma u_3 \in
  \llbracket A \rrbracket$. Hence, by Proposition~\ref{elimsup},
  $\elimsup(\sigma u_1,\abstr{x} \sigma u_2,\abstr{y} \sigma u_3)
  \in \llbracket A \rrbracket$, that is $\sigma t \in \llbracket A
  \rrbracket$.
  \qedhere
\end{itemize}
\end{proof}

\begin{corollary}[Termination]
  Let $t$ be a proof of $A$ in a context $\Gamma$. Then,
  $t$ strongly terminates.
\end{corollary}

\begin{proof}
  Let $\sigma$ be the substitution mapping each variable $x_i:A_i$ of
  $\Gamma$ to
  itself. Note that, by Proposition~\ref{Var}, this variable is an
  element of $\llbracket A_i \rrbracket$.  Then, $t = \sigma t$ is an
  element of $\llbracket A \rrbracket$. Hence, it strongly terminates.
\end{proof}

\subsection{Introduction property}

\begin{theorem}[Introduction]
\label{introductions}
Let $t$ be a closed irreducible proof of $A$.
\begin{itemize}
\item  If $A$ has the form $\top$, then $t$ is $\star$.

\item  The proposition $A$ is not $\bot$.
  
\item If $A$ has the form $B \Rightarrow C$, then $t$ has the form
  $\lambda \abstr{x}u$.

\item If $A$ has the form $B \wedge C$, then $t$ has the form
  $\pair{u}{v}$.

\item If $A$ has the form $B \vee C$, then $t$ has the form
  $\inl(u)$, $\inr(u)$, $u \plus v$, or $\bullet u$.

\item If $A$ has the form $B \odot C$, then $t$ has the form $\super{u}{v}$.
\end{itemize}
\end{theorem}

\begin{proof}
By induction on the structure of $t$.

We first remark that, as the proof $t$ is closed, it is not a
variable. Then, we prove that it cannot be an elimination.
\begin{itemize}

\item If $t = \elimtop(u,v)$, then $u$ is a closed
  irreducible proof of $\top$, hence, by induction
  hypothesis, it is $\star$
  and the proof $t$ is reducible.

\item If $t = \elimbot(u)$, then $u$ is a closed irreducible proof of
  $\bot$ and, by induction hypothesis, no such proof exists.

\item If $t = u~v$, then $u$ is a closed  irreducible proof of $B
  \Rightarrow A$, hence, by induction hypothesis, it has the form
  $\lambda \abstr{x}u_1$ and the proof $t$ is reducible.

\item If $t = \elimand^1(u,\abstr{x}v)$, then $u$ is a closed
  irreducible proof of $B \wedge C$, hence, by induction
  hypothesis, it has the form $\pair{u_1}{u_2}$
  and the proof $t$ is reducible.

\item If $t = \elimand^2(u,\abstr{x}v)$, then $u$ is a closed
  irreducible proof of $B \wedge C$, hence, by induction
  hypothesis, it has the form $\pair{u_1}{u_2}$
  and the proof $t$ is reducible.

\item If $t = \elimor(u,\abstr{x}v,\abstr{y}w)$, then $u$ is a closed
  irreducible proof of $B \vee C$, hence, by induction hypothesis, it
  has the form $\inl(u_1)$, $\inr(u_1)$, $u_1 \plus u_2$, or $\bullet
  u_1$ and the proof $t$ is reducible.

\item If $t = \elimsup(u,\abstr{x}v,\abstr{y}w)$, $t =
  \elimsup^1(u,\abstr{x}v)$, or $t = \elimsup^2(u,\abstr{x}v)$, then
  $u$ is a closed irreducible proof of $B \odot C$, hence, by
  induction hypothesis, it has the form $\super{u_1}{u_2}$ and the
  proof $t$ is reducible.
\end{itemize}

Hence, $t$ is an introduction, a sum, or a product.

It $t$ is $\star$, then $A$ is $\top$.  If it has the form $\lambda
\abstr{x}u$, then $A$ has the form $B \Rightarrow C$.  If it has the
form $\pair{u}{v}$, then $A$ has the form $B \wedge C$.  If it has the
form $\inl(u)$ or $\inr(u)$, then $A$ has the form $B \vee C$.  If it
has the form $\super{u}{v}$ then $A$ has the form $B \odot C$.  We
prove that, if it has the form $u \plus v$ or $\bullet u$, $A$ has
the form $B \vee C$.

\begin{itemize}
  \item 
If $t = u \plus v$, then 
the proofs $u$ and $v$ are two closed and irreducible proofs of
$A$. If $A = \top$ then, by induction hypothesis, they are both $\star$
and the proof $t$ is reducible.  If $A = \bot$ then, they are
closed irreducible proofs of $\bot$ and, by induction hypothesis, no such
proofs exist.  If $A$ has the form $B \Rightarrow C$ then, by
induction hypothesis, they are both abstractions and the proof $t$ is
reducible.  If $A$ has the form $B \wedge C$, then, by induction
hypothesis, they are both pairs and the proof $t$ is reducible.
If $A$ has the form $B \odot C$, then, by induction hypothesis, they
are both superpositions and the proof $t$ is reducible.  Hence,
$A$ has the form $B \vee C$.

\item
If $t = \bullet u$, then 
the proof $u$ is a closed and irreducible proof of
$A$. If $A = \top$ then, by induction hypothesis, $u$ is $\star$
and the proof $t$ is reducible.  If $A = \bot$ then, it is 
a closed irreducible proof of $\bot$ and, by induction hypothesis, no such
proof exists.  If $A$ has the form $B \Rightarrow C$ then, by
induction hypothesis, it is an abstraction and the proof $t$ is
reducible.  If $A$ has the form $B \wedge C$, then, by induction
hypothesis, it is a pair and the proof $t$ is reducible.
If $A$ has the form $B \odot C$, then, by induction
hypothesis, it is a superposition and the proof $t$ is reducible.
Hence,
$A$ has the form $B \vee C$.
\qedhere
\end{itemize}
\end{proof}

Note that we reap here the benefit of commuting, when possible, the
interstitial rules with the introduction rules, as, except for the
disjunction, closed irreducible proofs are genuine introductions.

\begin{proposition}[Disjunction]
If the proposition $A \vee B$ has a closed proof, then $A$ has a
closed proof or $B$ has a closed proof.
\end{proposition}

\begin{proof}
  Consider a closed proof of $A \vee B$ and its irreducible form $t$.
  We prove, by induction on the structure of $t$, that $A$ has a
  closed proof or $B$ has a closed proof. By Theorem~\ref{introductions}, $t$
  has the form $\inl(u)$, $\inr(u)$, $u \plus
  v$, or $\bullet u$.  If it has the form $\inl(u)$, $u$ is a closed
  proof of $A$. If it has the form $\inr(u)$, $u$ is a closed proof of
  $B$. If it has the form $u \plus v$ or $\bullet u$, $u$ is a closed
  irreducible proof of $A \vee B$. Thus, by induction hypothesis, $A$
  has a closed proof or $B$ has a closed proof.
\end{proof}

\section{Quantifying non-determinism}
\label{secquantifying}

When we have a non-deterministic reduction system, we often want to
quantify the propensity of a proof to reduce to another.

One way to do so is to consider a field $({\mathcal S},+,\times)$ of {\it
  scalars}, for instance ${\mathbb Q}$, ${\mathbb R}$, or ${\mathbb
  C}$, with the usual sum $+$ and product $\times$,
take a different rule $\top$-i for each scalar
and a different rule prod for each scalar.  So, for each scalar $a$,
we have a closed irreducible proof of $\top$ and we write $a.\star$
for this proof.  In the same way, we write $a \bullet t$ for the proof
obtained by applying, to the proof $t$, the rule prod corresponding to
the scalar $a$.  As the closed irreducible proofs of $\top$ are in
one-to-one correspondence with the elements of ${\mathcal S}$, those
of $\top \odot \top$ are in one-to-one with the elements of ${\mathcal
  S}^2$, those of $(\top \odot \top) \odot (\top \odot \top)$ are in
one-to-one correspondence with the elements of ${\mathcal S}^4$, etc.

In the $\odot$-calculus, the proof $\star \plus \star$ reduces to
$\star$.  Now, the proof ${a.\star} \plus b.\star$ reduces to $(a +
b).\star$, where the scalars are added.  In the same way the proof $a
\bullet b.\star$ reduces to $(a \times b).\star$, where the scalars
are multiplied.  In the $\odot$-calculus, the proof
$\elimtop(\star,t)$ reduces to $\bullet t$.  Now, the proof
$\elimtop(a.\star,t)$ reduces neither to $t$, as it does in the usual
intuitionistic propositional logic, nor to $\bullet t$ as it does in
the system of Section~\ref{seclogic}, but to $a \bullet t$.

\subsection{The \texorpdfstring{$\odot^{\mathcal S}$}{O-S}-calculus}

We define the $\odot^{\mathcal S}$-calculus (read: ``the sup-S-calculus''), by
extending the grammar of proofs as follows
\begin{align*}
  t =~& x \mid t \plus u \mid a \bullet t 
        \mid a.\star \mid \elimtop(t,u) \mid \elimbot(t)\\
      &\mid \lambda \abstr{x}t\mid t~u
        \mid \pair{t}{u} \mid 
\elimand^1(t,\abstr{x}u) \mid
\elimand^2(t,\abstr{x}u)\\
      &\mid \inl(t)\mid \inr(t)\mid \elimor(t,\abstr{x}u,\abstr{y}v)\\
&\mid \super{t}{u} \mid \elimsup(t,\abstr{x}u,\abstr{y}v)
\mid \elimsup^1(t,\abstr{x}u)
\mid \elimsup^2(t,\abstr{x}u)
\end{align*}
where $a$ is a scalar.

The typing rules are similar to those
of Figure~\ref{figtypingrules} except the rule
$$\irule{}
        {\Gamma \vdash \star:\top}
        {\mbox{$\top$-i}}$$
which is replaced with
$$\irule{}
        {\Gamma \vdash a.\star:\top}
        {\mbox{$\top$-i} (a)}$$
and the rule
$$\irule{\Gamma \vdash t:A}
        {\Gamma \vdash \bullet t:A}
        {\mbox{prod}}$$
which is replaced with
$$\irule{\Gamma \vdash t:A}
        {\Gamma \vdash a \bullet t:A}
        {\mbox{prod}(a)}$$

The reduction rules are those of Figure~\ref{figureductionrulesscalar}.  Note
that the scalar computation is
implicit and the terms expressing scalars are just constants. For
instance $2.\star \plus 3.\star$ reduces directly to $5.\star$.

The $\odot^{\mathcal S}$-calculus is thus a $\lambda$-calculus
equipped with a notion of linear combination of terms, such as Lineal
\cite{Lineal}, the Algebraic $\lambda$-calculus \cite{Vaux2009}, etc.

\begin{figure}[t]
  \[
    \begin{array}{r@{\,}l}
      \elimtop(a.\star,t) & \longrightarrow  a \bullet t\\
      (\lambda \abstr{x}t)~u & \longrightarrow  (u/x)t\\
      \elimand^1(\pair{t}{u}, \abstr{x}v) & \longrightarrow  (t/x)v\\
      \elimand^2(\pair{t}{u}, \abstr{x}v) & \longrightarrow  (u/x)v\\
      \elimor(\inl(t),\abstr{x}v,\abstr{y}w) & \longrightarrow  (t/x)v\\
      \elimor(\inr(u),\abstr{x}v,\abstr{y}w) & \longrightarrow  (u/y)w\\
      \elimsup(\super{t}{u},\abstr{x}v,\abstr{y}w) & \longrightarrow  (t/x)v\\
      \elimsup(\super{t}{u},\abstr{x}v,\abstr{y}w) & \longrightarrow  (u/y)w\\
      \elimsup^1(\super{t}{u},\abstr{x}v) & \longrightarrow  (t/x)v\\
      \elimsup^2(\super{t}{u},\abstr{x}v) & \longrightarrow  (u/x)v\\
      \\
      {a.\star} \plus b.\star&\longrightarrow  (a+b).\star\\
      (\lambda \abstr{x}t) \plus (\lambda \abstr{x}u) & \longrightarrow  \lambda \abstr{x}(t \plus u)\\
      \pair{t}{u} \plus \pair{v}{w}
      & \longrightarrow  \pair{t \plus v}{u \plus w} 
      \\
      \elimor(t \plus u,\abstr{x}v,\abstr{y}w) & \longrightarrow 
      \elimor(t,\abstr{x}v,\abstr{y}w)
      \plus
      \elimor(u,\abstr{x}v,\abstr{y}w)\\
      \super{t}{u} \plus \super{v}{w}
      & \longrightarrow  \super{t \plus v}{u \plus w} 
      \\
      \\
      a \bullet b.\star&\longrightarrow  (a \times b).\star\\
      a \bullet \lambda \abstr{x} t &\longrightarrow  \lambda \abstr{x} a \bullet t\\
      a \bullet \pair{t}{u} &\longrightarrow  \pair{a \bullet t}{a \bullet u}\\

      \elimor(a \bullet t,\abstr{x}v,\abstr{y}w) & \longrightarrow 
      a \bullet \elimor(t,\abstr{x}v,\abstr{y}w)\\
      a \bullet \super{t}{u} 
      & \longrightarrow  \super{a \bullet t}{a \bullet u} 
    \end{array}
  \]
  \caption{The reduction rules of the
$\odot^{\mathcal S}$-calculus \label{figureductionrulesscalar}}
\end{figure}

\subsection{Properties}

\begin{theorem}[Subject reduction]
  If $\Gamma \vdash t:A$ and $t \lra u$, then $\Gamma \vdash u:A$.
\end{theorem}

\begin{proof}
We first prove a substitution property and then proceed by induction
on the definition of the relation $\lra$.
\end{proof}

\begin{theorem}[Confluence]
This system of Figure~\ref{figureductionrulesscalar} without the rules
$$\elimsup(\super{t}{u},\abstr{x}v,\abstr{y}w) \longrightarrow  (t/x)v$$
$$\elimsup(\super{t}{u},\abstr{x}v,\abstr{y}w) \longrightarrow  (u/y)w$$
is confluent.
\end{theorem}

\begin{proof}
This system also is left linear and it has no critical pairs. Thus, by 
\cite[Theorem 6.8]{Nipkow} it is confluent.
\end{proof}

\begin{theorem}[Termination]
Let $t$ be a proof of $A$ in a context $\Gamma$. Then, $t$ strongly
terminates.
\end{theorem}

\begin{proof}
Consider a translation $^\circ$ of proofs from the $\odot^{\mathcal
  S}$-calculus to the $\odot$-calculus obtained by replacing the rules
$\mbox{$\top$-i}(a)$ with the rule $\top$-i and the rules
$\mbox{prod}(a)$ with the rule prod: $(a.\star)^\circ = \star$, $(a
\bullet t)^\circ = \bullet t^\circ$, etc.  If $t \longrightarrow u$ in
the $\odot^{\mathcal S}$-calculus, then $t^\circ \longrightarrow
u^\circ$ in the $\odot$-calculus.  Hence, the reduction in the
$\odot^{\mathcal S}$-calculus terminates.
\end{proof}

\begin{theorem}[Introduction]
Let $t$ be a closed irreducible proof of $A$.
\begin{itemize}
\item  If $A$ has the form $\top$, then $t$ is $a.\star$.

\item  The proposition $A$ is not $\bot$.
  
\item If $A$ has the form $B \Rightarrow C$, then $t$ has the form
  $\lambda \abstr{x}u$.

\item If $A$ has the form $B \wedge C$, then $t$ has the form
  $\pair{u}{v}$.

\item If $A$ has the form $B \vee C$, then $t$ has the form
  $\inl(u)$, $\inr(u)$, $u \plus v$, or $a \bullet u$.

\item If $A$ has the form $B \odot C$, then $t$ has the form $\super{u}{v}$.
\end{itemize}
\end{theorem}

\begin{proof}
Similar to that of Theorem~\ref{introductions}.
\end{proof}

\subsection{Quantifying non-determinism}

When ${\mathcal S}$ is ${\mathbb Q}$, ${\mathbb R}$, or ${\mathbb C}$, we can use
the scalars $a$ and $b$ to assign probabilities to the reductions
\[
  \elimsup(\super{t}{u},\abstr{x}v,\abstr{y}w)  \longrightarrow (t/x)v
\]
and
\[
\elimsup(\super{t}{u},\abstr{x}v,\abstr{y}w) \longrightarrow (u/y)w
\]

\begin{example}
\label{strategy}
We define a strategy where the rules
\[
  \elimsup(\super{t}{u},\abstr{x}v,\abstr{y}w)  \longrightarrow (t/x)v
\]
and
\[
\elimsup(\super{t}{u},\abstr{x}v,\abstr{y}w) \longrightarrow (u/y)w
\]
apply only when $t$ and $u$ are closed irreducible proofs.

In this case, if $a$ and $b$ are not both $0$, we assign the probabilities
$\tfrac{|a|^2}{|a|^2 + |b|^2}$ and $\tfrac{|b|^2}{|a|^2 + |b|^2}$ 
to the reductions
\[
  \elimsup(\super{a.\star}{b.\star},\abstr{x}v,\abstr{y}w)  \longrightarrow (a.\star/x)v
\]
and
\[
\elimsup(\super{a.\star}{b.\star},\abstr{x}v,\abstr{y}w) \longrightarrow (b.\star/y)w
\]
And if either $a = b = 0$ or $t$ and $u$ are proofs of propositions
different from $\top$, we assign any probability, for instance
$\frac{1}{2}$, to these reductions.
\end{example}

\section{Application to quantum computing}
\label{secquantum}

We now show that the $\odot^{\mathbb C}$-calculus, with the reduction
strategy of Example~\ref{strategy}, restricting the reduction of
$\elimsup(\super{t}{u},\abstr{x}v,\abstr{y}w)$ to the cases where $t$
and $u$ are closed irreducible proofs, contains the core of a small
quantum programming language.

\subsection{Bits}

\begin{definition}[Bit]
Let $\B = \top \vee \top$. The proofs
${\bf 0} = \inl(1.\star)$ and ${\bf 1} = \inr(1.\star)$ are closed
irreducible proofs of $\B$.
\end{definition}

Note that the proofs $\inl(1.\star)$ and $\inr(1.\star)$ are not the only
closed irreducible proofs of $\B$, for example $\inl(2.\star)$ and
$\inl(1.\star) \plus \inr(1.\star)$ also are.

\begin{definition}[Test]
The test operator is defined as
$$\test(t,u,v) = \elimor(t,\abstr{x}u,\abstr{y}v)$$
where $x$ and $y$ are variables not occurring in $u$ and $v$.  Note that 
$\test({\bf 0},u,v) \longrightarrow u$ and $\test({\bf 1},u,v)
\longrightarrow v$.
\end{definition}

\subsection{Qubits}

A $n$-qubit, for $n \geq 1$, is a vector of ${\mathbb C}^{2^n}$ of
norm $1$.  We show now how $n$-qubits, and more generally vectors of
${\mathbb C}^{2^n}$, for $n \geq 0$, can be expressed as proofs in the
$\odot^{\mathbb C}$-calculus.

\begin{definition}[The proposition $\Q^{\otimes n}$]
The proposition $\Q^{\otimes n}$ is defined by induction on $n$ as follows
  \begin{itemize}
  \item $\Q^{\otimes 0} = \top$,
  \item $\Q^{\otimes n+1} = \Q^{\otimes n} \odot \Q^{\otimes n}$.
    \end{itemize}
\end{definition}

Note that, in this definition, the binary connective $\odot$ is always
used with two identical propositions: $A \odot A$.

The proposition $\Q^{\otimes 1} = \top \odot \top$ is sometimes written $\Q$.

The closed irreducible proofs of $\Q^{\otimes n}$ and
the vectors of ${\mathbb C}^{2^n}$
are in one-to-one
correspondence: to each closed irreducible proof $t$ of
$\Q^{\otimes n}$, we
associate a vector $\underline{t}$ of ${\mathbb C}^{2^n}$ and to each
vector ${\bf u}$ of ${\mathbb C}^{2^n}$, we associate a closed irreducible
proof $\overline{\bf u}$ of $\Q^{\otimes n}$.

\begin{definition}[One-to-one correspondance]
To each closed irreducible proof $t$ of $\Q^{\otimes n}$,  we
associate a vector $\underline{t}$ of ${\mathbb C}^{2^n}$ as follows.
\begin{itemize}
\item
  If $n = 0$, then $t = a.\star$. We let $\underline{t} =
\left(\begin{smallmatrix} a \end{smallmatrix}\right)$.

\item
  If $n = n' + 1$, then $t = \super{u}{v}$.
  We let $\underline{t}$ be the vector with two
blocks $\underline{u}$ and $\underline{v}$:  $\underline{t} =
\left(\begin{smallmatrix}
  \underline{u}\\\underline{v} \end{smallmatrix}\right)$.
\end{itemize}

To each vector ${\bf u}$ of ${\mathbb C}^{2^n}$, we associate a closed
irreducible proof $\overline{\bf u}$ of $\Q^{\otimes n}$.

\begin{itemize}
\item If $n = 0$, then ${\bf u} =
  \left(\begin{smallmatrix} a \end{smallmatrix}\right)$.
  We let $\overline{\bf u} = a.\star$.

\item If $n = n' + 1$, 
let ${\bf u}_1$ and ${\bf u}_2$ be
  the two blocks of ${\bf u}$ of $2^{n'}$ lines, so ${\bf u} =
  \left(\begin{smallmatrix} {\bf u}_1\\ {\bf
      u}_2\end{smallmatrix}\right)$.  We let
    $\overline{\bf u} = \super{\overline{{\bf u}_1}}{\overline{{\bf
          u}_2}}$.
\end{itemize}
\end{definition}

In particular, the proof $0_{\Q^{\otimes n}}$ is defined as $\overline{\bf 0}$,
where ${\bf 0}$ is the zero vector of ${\mathbb C}^{2^n}$.

\begin{example}
The 2-qubit $\ket{01} = \left(\begin{smallmatrix}
  0\\ 1\\ 0\\ 0 \end{smallmatrix}\right)$ is expressed as the proof
$\overline{\ket{01}} = \super{\super{0.\star}{1.\star}}{\super{0.\star}{0.\star}}$
and the entangled 2-qubit
$\frac{1}{\sqrt{2}}.\ket{00} +
\frac{1}{\sqrt{2}}.\ket{11} = 
\left(\begin{smallmatrix}
  \frac{1}{\sqrt{2}} \\ 0\\ 0\\
\frac{1}{\sqrt{2}} \end{smallmatrix}\right)$ 
as the proof 
$\overline{\frac{1}{\sqrt{2}}.\ket{00} + \frac{1}{\sqrt{2}}.\ket{11}}
=
\super{\super{\frac{1}{\sqrt{2}}.\star}{0.\star}}{\super{0.\star}{\frac{1}{\sqrt{2}}.\star}}$.
\end{example}

We extend the definition of $\underline{t}$ to any closed proof of
$\Q^{\otimes n}$, $\underline{t}$ is by definition $\underline{t'}$,
where $t'$ is the irreducible form of $t$.

We also take the convention that any closed irreducible proof $u$ of
$\Q^{\otimes n}$, expressing a non-zero vector $\underline{u} \in
{\mathbb C}^{2^n}$, is an alternative expression of the $n$-qubit
$\frac{\underline{u}}{\|\underline{u}\|}$.  For example, the qubit
$\frac{1}{\sqrt{2}}.\ket{0} + \frac{1}{\sqrt{2}}.\ket{1}$ is
expressed as the proof
$\super{\frac{1}{\sqrt{2}}.\star}{\frac{1}{\sqrt{2}}.\star}$, but also
as the proof $\super{1.\star}{1.\star} = \overline{\ket{0} +
  \ket{1}}$.

The next Propositions show that the symbol $\plus$ expresses the sum of
vectors and the symbol $\bullet$, the product of a vector by a scalar.

\begin{proposition}[Sum of vectors]
  \label{parallelsum}
  Let $u$ and $v$ be two closed proofs of $\Q^{\otimes n}$. 
  Then, $\underline{u \plus v} = \underline{u} + \underline{v}$.
\end{proposition}

\begin{proof}
By induction on $n$. 

\begin{itemize}
  \item 
If $n = 0$, then $u \lra^* a.\star$, $v \lra^* b.\star$, $\underline{u} =
\left(\begin{smallmatrix} a \end{smallmatrix}\right)$, $\underline{v}
= \left(\begin{smallmatrix} b \end{smallmatrix}\right)$.  Thus,
$\underline{u \plus v} = \underline{{a.\star} \plus b.\star} =
\underline{(a + b).\star} = \left(\begin{smallmatrix} a +
  b \end{smallmatrix}\right) = \left(\begin{smallmatrix}
  a \end{smallmatrix}\right) + \left(\begin{smallmatrix}
  b \end{smallmatrix}\right) = \underline {u} + \underline{v}$.

\item 
  If $n = n' + 1$, then
  $u \lra^* \super{u_1}{u_2}$,  $v \lra^* \super{v_1}{v_2}$,
$\underline{u} = \left(\begin{smallmatrix} \underline{u_1} \\ \underline{u_2} 
 \end{smallmatrix}\right)$ and 
$\underline{v} = \left(\begin{smallmatrix} \underline{v_1} \\ \underline{v_2} 
  \end{smallmatrix}\right)$.
Thus, using the induction hypothesis,
$\underline{u \plus v} = 
\underline{\super{u_1}{u_2}
\plus 
\super{v_1}{v_2}}
=
\underline{\super{u_1 \plus v_1}{u_2 \plus v_2}}
=
\left(\begin{smallmatrix} \underline{u_1 \plus v_1} \\
                          \underline{u_2 \plus v_2}
\end{smallmatrix}\right)
= 
\left(\begin{smallmatrix} \underline{u_1} + \underline{v_1} \\
                          \underline{u_2} + \underline{v_2}
\end{smallmatrix}\right)
=
\left(\begin{smallmatrix} \underline{u_1} \\
                          \underline{u_2}
\end{smallmatrix}\right)
+
\left(\begin{smallmatrix} \underline{v_1} \\
                          \underline{v_2}
\end{smallmatrix}\right) = 
  \underline{u} + \underline{v}$.
\qedhere
\end{itemize}
\end{proof}

\begin{proposition}[Product of a vector by a scalar]
  \label{parallelprod}
  Let $u$ be a closed proof of $\Q^{\otimes n}$. 
  Then, $\underline{a \bullet u} = a \underline{u}$.
\end{proposition}

\begin{proof}
By induction on $n$.

\begin{itemize}
  \item 
If $n = 0$, then $u \lra^* b.\star$, $\underline{u} =
\left(\begin{smallmatrix} b \end{smallmatrix}\right)$.
Thus,
$\underline{a \bullet u} = \underline{a \bullet b.\star} =
\underline{(a \times b).\star} = \left(\begin{smallmatrix} a \times
  b \end{smallmatrix}\right) = 
  a \left(\begin{smallmatrix} b \end{smallmatrix}\right) = a \underline {u}$.

\item 
If $n = n' + 1$, $u \lra^* \super{u_1}{u_2}$,  
$\underline{u} = \left(\begin{smallmatrix} \underline{u_1} \\ \underline{u_2} 
 \end{smallmatrix}\right)$.
Thus, using the induction hypothesis,
$\underline{a \bullet u} = 
\underline{a \bullet \super{u_1}{u_2}}
=
\underline{\super{a \bullet u_1}{a \bullet u_2}}
=
\left(\begin{smallmatrix} \underline{a \bullet u_1} \\
                          \underline{a \bullet u_2}
\end{smallmatrix}\right)
= 
\left(\begin{smallmatrix} a \underline{u_1}  \\
                          a \underline{u_2} 
\end{smallmatrix}\right)
=
a \left(\begin{smallmatrix} \underline{u_1} \\
                          \underline{u_2}
\end{smallmatrix}\right)
= a \underline{u}$.
\qedhere
\end{itemize}
\end{proof}

\subsection{Matrices}

The information-preserving, reversible, and deterministic unitary
operators are expressed with the proof constructors $\elimsup^1$ and
$\elimsup^2$.

\begin{theorem}[Matrices]
\label{matrices}
Let $M$ be a matrix with $2^m$ columns and $2^n$ lines, then there
exists a closed proof $t$ of $\Q^{\otimes m} \Rightarrow \Q^{\otimes
  n}$ such that, for all vectors ${\bf u} \in {\mathbb C}^{2^m}$,
$\underline{t~\overline{\bf u}} = M {\bf u}$.
\end{theorem}

\begin{proof}
By induction on $A$.
\begin{itemize}

\item If $m = 0$, then $M$ is a matrix of one column and
  $2^n$ lines. Hence, it is also a vector of $2^n$ lines.
  We take
  $$t = \lambda \abstr{x} \elimtop(x,\overline{M})$$
  Let ${\bf u} \in {\mathcal S}^1$, ${\bf u}$ has the form
  $\left(\begin{smallmatrix}  a \end{smallmatrix}\right)$ and 
  $\overline{\bf u} = a.\star$.

  Then, using
  Proposition~\ref{parallelprod}, we have $\underline{t~\overline{\bf u}} 
= \underline{\elimtop(\overline{\bf u},\overline{M})}
= \underline{\elimtop(a.\star,\overline{M})}
= \underline{a \bullet \overline{M}}
= a \underline{\overline{M}} = a M = M
\left(\begin{smallmatrix}  a\end{smallmatrix}\right) =
  M {\bf u}$.

\item If $m = m'+1$, then let $M_1$ and $M_2$ be the two blocks of $M$
  of $2^{m'}$ columns, so $M = \left(\begin{smallmatrix} M_1 &
    M_2\end{smallmatrix}\right)$.

    By induction hypothesis, there exist closed proofs $t_1$ and $t_2$ of
    the proposition $\Q^{m'} \Rightarrow \Q^n$ such that, for all
    vectors ${\bf u}_1, {\bf u}_2 \in {\mathbb C}^{2^{m'}}$, we have
    $\underline{t_1~\overline{{\bf u}_1}} = M_1 {\bf u}_1$ and
    $\underline{t_2~\overline{{\bf u}_2}}= M_2 {\bf u}_2$.  We take
    $$t = \lambda \abstr{x} (\elimsup^1(x, \abstr{y} (t_1~y)) \plus \elimsup^2(x,
    \abstr{z} (t_2~z)))$$
    Let ${\bf u} \in {\mathbb C}^{2^m}$, and ${\bf u}_1$ and ${\bf
      u}_2$ be the two blocks of $2^{m'}$ lines of ${\bf u}$, so ${\bf
      u} = \left(\begin{smallmatrix} {\bf u}_1 \\ {\bf
        u}_2 \end{smallmatrix}\right)$, and $\overline{\bf u} =
    \super{\overline{{\bf u}_1}}{\overline{{\bf u}_2}}$.

 Then, using Proposition~\ref{parallelsum}, $\underline{t~\overline{\bf
     u}} = \underline{ \elimsup^1(\super{\overline{{\bf
         u}_1}}{\overline{{\bf u}_2}}, \abstr{y}
   (t_1~y)) \plus \elimsup^2(\super{\overline{{\bf
         u}_1}}{\overline{{\bf u}_2}}, \abstr{z}
   (t_2~z))}
   = \underline{t_1~\overline{{\bf u}_1} \plus t_2~\overline{{\bf u}_2}}
   = \underline{t_1~\overline{{\bf u}_1}} + \underline{t_2~\overline{{\bf u}_2}}
   = M_1 {\bf u}_1 + M_2 {\bf u}_2
   = \left(\begin{smallmatrix} M_1 & M_2 \end{smallmatrix}\right)
   \left(\begin{smallmatrix} {\bf u}_1 \\ {\bf u}_2  \end{smallmatrix}\right)
   = M {\bf u}$.
 \qedhere
\end{itemize}
\end{proof}

\begin{example}[Matrices with two colums and two lines]
The matrix
$\left(\begin{smallmatrix} a & c\\b & d \end{smallmatrix}\right)$
is expressed as the proof
$$t = \lambda \abstr{x} (\elimsup^1(x,\abstr{y}
\elimtop(y,\super{a.\star}{b.\star})) \plus \elimsup^2(x,\abstr{z}
\elimtop(z,\super{c.\star}{d.\star})))$$
Then
\begin{align*}
t~\super{e.\star}{f.\star} & \lra
\elimsup^1(\super{e.\star}{f.\star},\abstr{y} \elimtop(y,\super{a.\star}{b.\star}))
\plus
\elimsup^2(\super{e.\star}{f.\star},\abstr{z} \elimtop(z,\super{c.\star}{d.\star}))\\
& \lra^* \elimtop(e.\star,\super{a.\star}{b.\star}) \plus
\elimtop(f.\star,\super{c.\star}{d.\star})\\
& \lra^* e \bullet \super{a.\star}{b.\star} \plus
f \bullet \super{c.\star}{d.\star}\\
& \lra^*
\super{(a \times e).\star}{(b \times e).\star}
\plus 
\super{(c \times f).\star}{(d \times f).\star}\\
& \lra^*
\super{(a \times e + c \times f).\star}{(b \times e + d \times f).\star}
\end{align*}
For instance, the Hadamard matrix 
$H = \left(\begin{smallmatrix} \tfrac{1}{\sqrt{2}} & \tfrac{1}{\sqrt{2}}\\ 
                               \tfrac{1}{\sqrt{2}} &\tfrac{-1}{\sqrt{2}}
 \end{smallmatrix}\right)$
is expressed as the proof
$$\lambda \abstr{x} \elimsup^1(x,\abstr{y}
\elimtop(y,\super{\tfrac{1}{\sqrt{2}}.\star}{\tfrac{1}{\sqrt{2}}.\star})
) \plus \elimsup^2(x,\abstr{z}
\elimtop(z,\super{\tfrac{1}{\sqrt{2}}.\star}{\tfrac{-1}{\sqrt{2}}.\star})
)$$
\end{example}

\subsection{Probabilities}

\begin{definition}[Norm]
  Let $t$ be a closed irreducible proof of $\Q^{\otimes n}$, we define the
  square of the norm $\|t\|^2$ of $t$ by induction on $n$.
  \begin{itemize}
    \item If $n = 0$, then $t = a.\star$ and we take $\|t\|^2 = |a|^2$.
\item  If $n = n'+1$, then $t = \super{u_1}{u_2}$ and we take
  $\|t\|^2 = \|u_1\|^2 + \|u_2\|^2$.
  \end{itemize}
\end{definition}

If $t$ is a closed irreducible proof of $\Q^{\otimes n}$ of the form
  $\super{u_1}{u_2}$, 
  where $\|u_1\|^2$ and $\|u_2\|^2$ are not both $0$,
  then we assign the probability 
  $\frac{\|u_1\|^2}{\|u_1\|^2 + \|u_2\|^2}$
  to the reduction 
  $$\elimsup(\super{u_1}{u_2},\abstr{x}v,\abstr{y}w) \longrightarrow (u_1/x)v$$
  and the probability 
  $\frac{\|u_2\|^2}{\|u_1\|^2 + \|u_2\|^2}$
  to the reduction 
$$\elimsup(\super{u_1}{u_2},\abstr{x}v,\abstr{y}w) \longrightarrow (u_2/y)w$$

If $\|u_1\|^2 = \|u_2\|^2 = 0$, or $u_1$ and $u_2$ are proofs of
propositions of a different form, we associate any probability, for
example $\tfrac{1}{2}$, to both reductions.

\begin{example}
If $t$ is a closed irreducible proof of $\Q$ of the form
$\super{a.\star}{b.\star}$, where $a$ and $b$ are not both $0$, then we assign
the probability
$\tfrac{|a|^2}{|a|^2 + |b|^2}$
to the reduction
$$\elimsup(\super{a.\star}{b.\star},\abstr{x}v,\abstr{y}w) \longrightarrow (a.\star/x)v$$
and
$\tfrac{|b|^2}{|a|^2 + |b|^2}$  to the reduction
$$\elimsup(\super{a.\star}{b.\star},\abstr{x}v,\abstr{y}w) \longrightarrow (b.\star/y)w$$

If $t$ is a closed irreducible proof of $\Q^{\otimes 2}$ of the form
$\super{\super{a.\star}{b.\star}}{\super{c.\star}{d.\star}}$ where $a$, $b$,
$c$, and $d$ are not all $0$, then we assign the probability 

$\tfrac{|a|^2 + |b|^2}{|a|^2 + |b|^2 + |c|^2 + |d|^2}$ to the reduction
$$\elimsup(\super{\super{a.\star}{b.\star}}{\super{c.\star}{d.\star}},\abstr{x}v,\abstr{y}w) \longrightarrow (\super{a.\star}{b.\star}/x)v$$
and $\tfrac{|c|^2 + |d|^2}{|a|^2 + |b|^2 + |c|^2 + |d|^2}$ to the reduction
$$\elimsup(\super{\super{a.\star}{b.\star}}{\super{c.\star}{d.\star}},\abstr{x}v,\abstr{y}w) \longrightarrow ((\super{c.\star}{d.\star})/y)w$$
\end{example}

\subsection{Measure}
\label{measurement}

\begin{figure}[t]
  \[
    \begin{array}{r@{\,}l}
      \pi_n &= \lambda \abstr{x} \elimsup(x, \abstr{y} \super{y}{0_{\Q^{\otimes n-1}}},\abstr{z} \super{0_{\Q^{\otimes n-1}}}{z})\\
\pi'_n &= \lambda \abstr{x} \elimsup(x, \abstr{y} {\bf 0}, \abstr{z} {\bf 1})\\
\pi''_n & = \lambda \abstr{x} \elimsup(x, \abstr{y}
\pair{\super{y}{0_{\Q^{\otimes n-1}}}}
     {{\bf 0}},
     \abstr{z} \pair{\super{0_{\Q^{\otimes n-1}}}{z}}
     {{\bf 1}})
    \end{array}
  \]
\caption{Measurement operators\label{figmesure}}
\end{figure}

The information-erasing, non-reversible, and non-deterministic
measurement operators are expressed with the proof constructor
$\elimsup$.

Several such operators are defined in Figure~\ref{figmesure}.
Let $n$ be a non-zero natural number and $t$ be a closed irreducible
proof of $\Q^{\otimes n}$ of the form $\super{u_1}{u_2}$, such that
$\|t\|^2 = \|u_1\|^2 + \|u_2\|^2 \neq 0$, expressing the state of an
$n$-qubit.  The proof $\pi_n~t$ of the proposition $\Q^{\otimes n}$
reduces, with probabilities $\tfrac{\|u_1\|^2}{\|u_1\|^2 + \|u_2\|^2}$
and $\tfrac{\|u_2\|^2}{\|u_1\|^2 + \|u_2\|^2}$ to
$\super{u_1}{0_{\Q^{\otimes n-1}}}$ and to $\super{0_{\Q^{\otimes
      n-1}}}{u_2}$.  It is the state of the $n$-qubit, after the partial
measure of the first qubit.  The proof $\pi'_n~t$ of the proposition
$\B$ reduces, with the same probabilities, to ${\bf 0}$ and to ${\bf
  1}$.  It is the ``classical'' result of the measure.  The proof
$\pi''_n~t$ of the proposition $\Q^{\otimes n} \wedge \B$ reduces,
with the same probabilities, to $\pair{\super{u_1}{0_{\Q^{\otimes
        n-1}}}}{{\bf 0}}$ and to $\pair{\super{0_{\Q^{\otimes
        n-1}}}{u_2}}{{\bf 1}}$.  It is the pair formed with the state of
the $n$-qubit, after the measure, and the ``classical'' result of the
measure.

\begin{example}
In the case $n = 1$, 
if $t$ is a closed irreducible proof of $\Q^{\otimes 1}$ of the form
$\super{a.\star}{b.\star}$, such that $a$ and $b$ are not both $0$, then the
proof $\pi_1~t$ of the proposition $\Q^{\otimes 1}$ reduces, with
probabilities $\tfrac{|a|^2}{|a|^2 + |b|^2}$ and $\tfrac{|b|^2}{|a|^2
  + |b|^2}$, to $\super{a.\star}{0.\star}$, that is an expression of $\ket{0}$,
and to $\super{0.\star}{b.\star}$, that is an expression of $\ket{1}$.
The proof $\pi'_1~t$ of the proposition $\B$ reduces, with
the same probabilities, to ${\bf 0}$ and to ${\bf 1}$.
The proof $\pi''_1~t$ of the proposition $\Q^{\otimes 1} \wedge \B$
reduces, with the same probabilities, to $\pair{\super{a.\star}{0.\star}}{{\bf
    0}}$ and to $\pair{\super{0.\star}{b.\star}}{{\bf 1}}$.
\end{example}

\begin{example}
  In the case $n = 2$, if $t$ is a closed irreducible proof of
  $\Q^{\otimes 2}$ of the form
  $\super{\super{a.\star}{b.\star}}{\super{c.\star}{d.\star}}$ where
  $a$, $b$, $c$, and $d$ are not all $0$, then the proof $\pi_2~t$ of
  the proposition $\Q^{\otimes 2}$ reduces, with probabilities
  $\tfrac{|a|^2 + |b|^2}{|a|^2 + |b|^2 + |c|^2 + |d|^2}$ and
  $\tfrac{|c|^2 + |d|^2}{|a|^2 + |b|^2 + |c|^2 + |d|^2}$, to
  $\super{\super{a.\star}{b.\star}}{\super{0.\star}{0.\star}}$ and to
  $\super{\super{0.\star}{0.\star}}{\super{c.\star}{d.\star}}$.  The
  proof $\pi'_2~t$ of the proposition $\B$ reduces, with the same
  probabilities, to ${\bf 0}$ and to ${\bf 1}$.  The proof $\pi''_2~t$
  of the proposition $\Q^{\otimes 2} \wedge \B$ reduces, with the same
  probabilities, to
  $\pair{\super{\super{a.\star}{b.\star}}{\super{0.\star}{0.\star}}}{{\bf
      0}}$ and to
  $\pair{\super{\super{0.\star}{0.\star}}{\super{c.\star}{d.\star}}}{{\bf
      1}}$.
\end{example}

Using the representation of matrices, it is possible to define
measurement operators that measure in any basis, by changing basis,
measuring, and changing basis again.  This way, it is also possible to
define measurement operators that partially measure, not the first
qubit of a $n$-qubit, but any.

\subsection{An example: Deutsch's algorithm}

Deutsch's algorithm allows to decide whether a 1-bit to 1-bit function $f$ is
constant or not, applying an oracle $U_f$, implementing $f$, only once. It is an
algorithm operating on 2-qubits. It proceeds in four steps.
\begin{enumerate}
\item 
Prepare the initial state
\(
\ket{+-} = \frac{1}{2} \ket{00} -\frac{1}{2} \ket{01}
+\frac{1}{2} \ket{10} -\frac{1}{2} \ket{11}
\).

\item 
Apply to it the unitary operator
$$U_f = \left(\begin{smallmatrix}
\mbox{\it if}(f 0,1,0) & \mbox{\it if}(f 0,0,1) & 0 & 0\\
\mbox{\it if}(f 0,0,1) & \mbox{\it if}(f 0,1,0) & 0 & 0\\
0 & 0 & \mbox{\it if}(f 1,1,0) & \mbox{\it if}(f 1,0,1)\\
0 & 0 & \mbox{\it if}(f 1,0,1) & \mbox{\it if}(f 1,1,0)
\end{smallmatrix}\right)$$
where $\mbox{\it if}(0,n,m) = n$ and $\mbox{\it if}(1,n,m) = m$.

Note that $U_f\ket{x,y} = \ket{x,y \oplus f(x)}$
  for $x,y\in\{0,1\}$, where $\oplus$ is the exclusive disjunction.

\item 
  Apply to it the unitary operator
  $$H\otimes I=
  \frac 1{\sqrt 2}\left( \begin{smallmatrix}
    1 & 0 & 1 & 0 \\
    0 & 1 & 0 & 1 \\
    1 & 0 & -1 & 0 \\
    0 & 1 & 0 & -1 
  \end{smallmatrix}\right)
  $$

\item
Measure the first qubit. The output is 
$0$,
if $f$ is constant and $1$ if it is not.
\end{enumerate}

In the $\odot^{\mathbb C}$-calculus, the initial state is
$$\overline{\ket{+-}} =
\super{\super{\frac{1}{2}.\star}{\frac{-1}{2}.\star}}
{\super{\frac{1}{2}.\star}{\frac{-1}{2}.\star}}$$
the function mapping the function $f$ to the operator $U_f$ is
expressed as in the proof of Theorem~\ref{matrices}
$$U = \lambda \abstr{f} \lambda \abstr{t}
(\elimsup^1(t,\abstr{x}
(
\elimsup^1(x,\abstr{z_0}M_0~z_0)
\plus
\elimsup^2(x,\abstr{z_1}M_1~z_1))
)
\plus
\elimsup^2(t,
\abstr{y}
(\elimsup^1(y,\abstr{z_2}~M_2~z_2)
\plus
\elimsup^2(y,\abstr{z_3}~M_3~z_3))
))
$$
with
$$M_0 = \lambda \abstr{s} \elimtop(s, \test(f~{\bf 0},
 \super{\super{1.\star}{0.\star}}{\super{0.\star}{0.\star}},
 \super{\super{0.\star}{1.\star}}{\super{0.\star}{0.\star}}))$$
$$M_1 = \lambda \abstr{s} \elimtop(s, \test(f~{\bf 0},
 \super{\super{0.\star}{1.\star}}{\super{0.\star}{0.\star}},
 \super{\super{1.\star}{0.\star}}{\super{0.\star}{0.\star}}))$$
$$M_2 = \lambda \abstr{s} \elimtop(s, \test(f~{\bf 1},
 \super{\super{0.\star}{0.\star}}{\super{1.\star}{0.\star}},
 \super{\super{0.\star}{0.\star}}{\super{0.\star}{1.\star}}))$$
$$M_3 = \lambda \abstr{s} \elimtop(s, \test(f~{\bf 1},
 \super{\super{0.\star}{0.\star}}{\super{0.\star}{1.\star}},
 \super{\super{0.\star}{0.\star}}{\super{1.\star}{0.\star}}))$$ 
The operator $H\otimes I$ is also expressed as
in the proof of Theorem~\ref{matrices}
and Deutsch's algorithm is the proof of
$(\B\Rightarrow\B)\Rightarrow\B$
$$\mbox{\it Deutsch}
=
\lambda \abstr{f}\pi'_2 ((H\otimes I)~(U~f~\overline{\ket{+-}}))$$
Let $f$ be a proof of $\B\Rightarrow\B$.
If $f$ is a constant function, we have $\mbox{\it
  Deutsch}\ f \lra^* {\bf 0}$, while if $f$ if not constant, $\mbox{\it Deutsch}\
f \lra^* {\bf 1}$.

\section{Conclusion}

We have extended intuitionistic propositional logic with a connective
$\odot$, that has both excessive and harmonious deduction rules, and
with interstitial rules.  We have then extended this logic again with
scalars.  We have shown that the proof language of this logic forms
the core of a quantum programming language.

The connective $\odot$, with its elimination symbol $\elimsup$, models
information-erasure, non-reversibility, and non-determinism, that
occur, for example, in quantum measurement.  With its elimination
symbols $\elimsup^1$ and $\elimsup^2$, it models the
information-preservation, reversibility, and determinism that occur,
for example, in unitary transformations.

There are several points that we did not address in this paper.
First, we leave open the question of the interpretation of this logic
in a model, in particular a categorical one, besides the obvious
Lindenbaum algebra.

Then, these notions of insufficient and excessive deduction rules are not
specific to natural deduction and similar notions could be defined and
investigated, for instance, in sequent calculus.
Note that in the sequent calculus, harmony can be defined in a stronger
sense, that includes, not only the possibility to reduce proofs, but
also to reduce the use of the rule axiom on non-atomic propositions to
smaller ones \cite{MillerPimentel}---an analogue of the $\eta$-expansion,
but generalized to arbitrary connectives.

Finally, the $\odot^{\mathbb C}$-calculus can express all quantum
circuits, as it can express matrices and measurement
operators. However, it is not restricted to quantum algorithms, since
the $\odot$ connective addresses the question of the
information-erasure, non-reversibility, and non-determinism of
measurement, but not that of linearity and unitarity.
We have started investigating, in \cite{fscd22}, 
a restriction of the $\odot$-calculus to linear operators,
forbidding, for example, the non-linear proof of the proposition $\Q
\Rightarrow \Q^{\otimes 2}$
\[
  \begin{array}{r@{}l@{}l}
    \lambda \abstr{x}
    &\elimsup^1(x,\, & \abstr{y}
    \elimsup^1(x, \abstr{y_1}\super{\super{\elimtop(y,y_1)}{0.\star}}
            {\super{0.\star}{0.\star}}) \plus \elimsup^2(x, \abstr{z_1} \super{\super{0.\star}{\elimtop(y,z_1)}}{\super{0.\star}{0.\star}}))\\
      & \plus\\
      &\elimsup^2(x,\,& \abstr{z}
      \elimsup^1(x, \abstr{y_2} \super{\super{0.\star}{0.\star}}{\super{\elimtop(z,y_2)}{0.\star}}) \plus \elimsup^2(x, \abstr{z_2} \super{\super{0.\star}{0.\star}}{\super{0.\star}{\elimtop(z,z_2)}}))
  \end{array}
\]
that maps $\super{a.\star}{b.\star}$ to
$\super{\super{a^2.\star}{ab.\star}}{\super{ab.\star}{b^2.\star}}$, that is
$a.\ket{0} + b.\ket{1}$ to
$a^2.\ket{00} + ab.\ket{01} +
ab.\ket{10} + b^2.\ket{11}$ and thus expresses cloning.

\section*{Acknowledgements}

The authors want to thank Jean-Baptiste Joinet, Jean-Pierre Jouannaud,
Dale Miller, Alberto Naibo, and Alex Tsokurov for useful discussions.

\bibliographystyle{elsarticle-harv}
\bibliography{../introelim}

\appendix

\section{Strong termination of proof reduction in
  intuitionistic propositional natural deduction}

\begin{definition}[Syntax]
  \begin{align*}
    A =~& \top \mid \bot \mid A \Rightarrow A \mid A \wedge A \mid A \vee A\\
    t =~& x\mid \star \mid \elimtop(t,u) \mid \elimbot(t)\mid \lambda \abstr{x}t\mid t~u\\
    &\mid \pair{t}{u} \mid \elimand^1(t,\abstr{x}u)\mid \elimand^2(t,\abstr{x}u)    \\
        &\mid \inl(t)\mid \inl(r)\mid \elimor(t,\abstr{x}u,\abstr{y}v)
  \end{align*}
\end{definition}

The proofs of the form $\star$, $\lambda \abstr{x}t$, $\pair{t}{u}$,
$\inl(t)$, and $\inr(t)$ are called {\it introductions}.

\begin{definition}[Typing rules]
$$\irule{}
        {\Gamma \vdash x:A}
        {\mbox{axiom~$x:A \in \Gamma$}}$$
$$\irule{}
        {\Gamma \vdash \star:\top}
        {\mbox{$\top$-i}}$$
$$\irule{\Gamma \vdash t:\top & \Gamma \vdash u:C}
        {\Gamma \vdash \elimtop(t,u):C}
        {\mbox{$\top$-e}}$$
$$\irule{\Gamma \vdash t:\bot}
        {\Gamma \vdash \elimbot(t):C}
        {\mbox{$\bot$-e}}$$
$$\irule{\Gamma, x:A \vdash t:B}
        {\Gamma \vdash \lambda \abstr{x}t:A \Rightarrow B}
        {\mbox{$\Rightarrow$-i}}$$
$$\irule{\Gamma \vdash t:A\Rightarrow B & \Gamma \vdash u:A}
        {\Gamma \vdash t~u:B}
        {\mbox{$\Rightarrow$-e}}$$
$$\irule{\Gamma \vdash t:A & \Gamma \vdash u:B}
        {\Gamma \vdash \pair{t}{u}:A \wedge B}
        {\mbox{$\wedge$-i}}$$
$$\irule{\Gamma \vdash t:A \wedge B & \Gamma, x:A \vdash u:C}
        {\Gamma \vdash \elimand^1(t,\abstr{x}u):C}
        {\mbox{$\wedge$-e1}}$$
$$\irule{\Gamma \vdash t:A \wedge B & \Gamma, x:B \vdash u:C}
        {\Gamma \vdash \elimand^2(t,\abstr{x}u):C}
        {\mbox{$\wedge$-e2}}$$
$$\irule{\Gamma \vdash t:A}
        {\Gamma \vdash \inl(t):A \vee B}
        {\mbox{$\vee$-i1}}$$
$$\irule{\Gamma \vdash t:B}
        {\Gamma \vdash \inr(t):A \vee B}
        {\mbox{$\vee$-i2}}$$
        $$\irule{\Gamma \vdash t:A \vee B & \Gamma, x:A \vdash u:C & \Gamma, y:B \vdash v:C}
        {\Gamma \vdash \elimor(t,\abstr{x}u,\abstr{y}v):C}
        {\mbox{$\vee$-e}}$$
\end{definition}  

\begin{definition}[Reduction rules]
  \label{app-reductionrules}
  \begin{align*}
    \elimtop(\star, v) & \longrightarrow v\\
    (\lambda \abstr{x}t)~u & \longrightarrow (u/x)t\\
    \elimand^1(\pair{t}{u}, \abstr{x}v) & \longrightarrow (t/x)v\\
    \elimand^2(\pair{t}{u}, \abstr{x}v) & \longrightarrow (u/x)v\\
    \elimor(\inl(t),\abstr{x}v,\abstr{y}w) & \longrightarrow (t/x)v\\
    \elimor(\inr(u),\abstr{x}v,\abstr{y}w) & \longrightarrow (u/y)w
  \end{align*}
\end{definition}

\begin{definition}
\label{app-reducibility}
  We define, by induction on the proposition $A$, a set of proofs
$\llbracket A \rrbracket$:
\begin{itemize}
\item $t \in \llbracket \top \rrbracket$ if $t$ strongly terminates,

\item $t \in \llbracket \bot \rrbracket$ if $t$ strongly terminates,

\item $t \in \llbracket A \Rightarrow B \rrbracket$ if $t$ strongly
  terminates and whenever it reduces to a proof of the form $\lambda
  \abstr{x}u$, then for every $v \in \llbracket A \rrbracket$, $(v/x)u \in
  \llbracket B \rrbracket$,

\item $t \in \llbracket A \wedge B \rrbracket$ if $t$ strongly
  terminates, whenever it reduces to a proof of the form $\pair{u}{v}$,
  then $u \in \llbracket A \rrbracket$ and $v \in \llbracket
  B \rrbracket$,

\item $t \in \llbracket A \vee B \rrbracket$ if $t$ strongly
  terminates, whenever it reduces to a proof of the form $\inl(u)$,
  then $u \in \llbracket A \rrbracket$, and whenever it reduces to a
  proof of the form $\inr(v)$, then $v \in \llbracket B \rrbracket$.
\end{itemize}
\end{definition}

\begin{definition}
If $t$ is a strongly terminating proof, we write $|t|$ for the
maximum length of a reduction sequence issued from $t$.
\end{definition}

\begin{proposition}[Variables]
\label{app-Var}
For any $A$, the set $\llbracket A \rrbracket$ contains all the variables.
\end{proposition}

\begin{proof}
A variable is irreducible, hence it strongly terminates. Moreover, it
never reduces to an introduction.
\end{proof}   

\begin{proposition}[Closure by reduction]
\label{app-closure}
If $t \in \llbracket A \rrbracket$ and $t \longrightarrow^* t'$, then 
$t' \in \llbracket A \rrbracket$.
\end{proposition}

\begin{proof}
If $t \longrightarrow^* t'$ and $t$ strongly terminates, then $t'$
strongly terminates.

Furthermore, if $A$ has the form $B \Rightarrow C$ and $t'$ reduces to
$\lambda \abstr{x}u$, then so does $t$, hence for every $v \in \llbracket B
\rrbracket$, $(v/x)u \in \llbracket C \rrbracket$.

If $A$ has the form $B \wedge C$ and $t'$ reduces to $\pair{u}{v}$,
then so does $t$, hence $u \in \llbracket B \rrbracket$ and $v \in
\llbracket C \rrbracket$.

If $A$ has the form $B \vee C$ and $t'$ reduces to $\inl(u)$, then so
does $t$, hence $u \in \llbracket B \rrbracket$ and if $A$ has the
form $B \vee C$ and $t'$ reduces to $\inr(v)$, then so does $t$, hence
$v \in \llbracket C \rrbracket$.
\end{proof}

\begin{proposition}[Girard's lemma]
\label{app-CR3}
Let $t$ be a proof that is not an introduction, 
such that all the one-step reducts of $t$
are in $\llbracket A \rrbracket$. Then, $t \in \llbracket A \rrbracket$.
\end{proposition}

\begin{proof}
Let $t, t_2, \ldots$ be a reduction sequence issued from $t$. If it has a
single element, it is finite. Otherwise, we have $t \longrightarrow
t_2$. As $t_2 \in \llbracket A \rrbracket$, it strongly terminates and
the reduction sequence is finite. Thus, $t$ strongly terminates.

Furthermore, if $A$ has the form $B \Rightarrow C$ and $t
\longrightarrow^* \lambda \abstr{x}u$, then let $t , t_2, \ldots, t_n$ be a
reduction sequence from $t$ to $\lambda \abstr{x}u$.  As $t_n$ is an
introduction and $t$ is not, $n \geq 2$. Thus, $t \longrightarrow t_2
\longrightarrow^* t_n$. We have $t_2 \in \llbracket A \rrbracket$,
thus for all $v \in \llbracket B \rrbracket$, $(v/x)u \in \llbracket C
\rrbracket$.

And if $A$ has the form $B \wedge C$ and $t \longrightarrow^* \pair{u}{v}$,
then let $t , t_2, \ldots, t_n$ be a reduction sequence
from $t$ to $\pair{u}{v}$.  As $t_n$ is an introduction and
$t$ is not, $n \geq 2$. Thus, $t \longrightarrow t_2 \longrightarrow^*
t_n$. We have $t_2 \in \llbracket A \rrbracket$, thus $u \in
\llbracket B \rrbracket$ and $v \in \llbracket C \rrbracket$.

If $A$ has the form $B \vee C$ and $t \longrightarrow^* \inl(u)$, then let $t ,
t_2, \ldots, t_n$ be a reduction sequence from $t$ to $\inl(u)$.  As
$t_n$ is an introduction and $t$ is not, $n \geq 2$.  Thus, $t
\longrightarrow t_2 \longrightarrow^* t_n$. We have $t_2 \in
\llbracket A \rrbracket$, thus $u \in \llbracket B \rrbracket$.

If $A$ has the form $B \vee C$ and $t \longrightarrow^* \inr(v)$, then let $t ,
t_2, \ldots, t_n$ be a reduction sequence from $t$ to $\inr(v)$.  As
$t_n$ is an introduction and $t$ is not, $n \geq 2$.  Thus, $t
\longrightarrow t_2 \longrightarrow^* t_n$. We have $t_2 \in
\llbracket A \rrbracket$, thus $v \in \llbracket C \rrbracket$.
\end{proof}

In Propositions~\ref{app-star} to~\ref{app-elimor}, we prove the
adequacy of each proof constructor. 

\begin{proposition}[Adequacy of $\star$]
\label{app-star}
We have $\star \in \llbracket \top \rrbracket$.
\end{proposition}

\begin{proof}
As $\star$ is irreducible, it strongly terminates, hence
$\star \in \llbracket \top \rrbracket$. 
\end{proof}

\begin{proposition}[Adequacy of $\lambda$]
\label{app-abstraction}
If, for all $u \in \llbracket A \rrbracket$, $(u/x)t \in \llbracket B
\rrbracket$, then $\lambda \abstr{x}t \in \llbracket A \Rightarrow B
\rrbracket$.
\end{proposition}

\begin{proof}
By Proposition~\ref{app-Var}, $x \in \llbracket A \rrbracket$, thus
$t = (x/x)t \in \llbracket B \rrbracket$. Hence, $t$ strongly
terminates.  Consider a reduction sequence issued from $\lambda
\abstr{x}t$.  This sequence can only reduce $t$ hence it is finite. Thus,
$\lambda \abstr{x}t$ strongly terminates.

Furthermore, if $\lambda \abstr{x}t \longrightarrow^* \lambda \abstr{x}t'$, then
$t \lra^* t'$.  Let $u \in \llbracket A \rrbracket$,
$(u/x)t \lra^* (u/x)t'$. 
By Proposition~\ref{app-closure}, $(u/x)t' \in \llbracket B \rrbracket$.
\end{proof}

\begin{proposition}[Adequacy of $\pair{}{}$]
\label{app-pair}
If $t_1 \in \llbracket A \rrbracket$ and $t_2 \in \llbracket B
\rrbracket$, then $\pair{t_1}{t_2} \in \llbracket A \wedge B
\rrbracket$.
\end{proposition}

\begin{proof}
The proofs $t_1$ and $t_2$ strongly terminate. Consider a reduction
sequence issued from $\pair{t_1}{t_2}$.  This sequence can only
reduce $t_1$
and $t_2$, hence it is finite.  Thus, $\pair{t_1}{t_2}$
strongly terminates.

Furthermore, if $\pair{t_1}{t_2} \longrightarrow^* \pair
{t'_1}{t'_2}$, then $t_1 \lra^* t'_1$ and $t_2 \lra^* t'_2$.  By
Proposition~\ref{app-closure}, $t'_1 \in \llbracket A \rrbracket$ and
$t'_2 \in \llbracket B \rrbracket$.
\end{proof}

\begin{proposition}[Adequacy of $\inl$]
\label{app-inl}
If $t \in \llbracket A \rrbracket$, then $\inl(t) \in \llbracket A
\vee B \rrbracket$.
\end{proposition}

\begin{proof}
The proof $t$ strongly terminates. Consider a reduction sequence
issued from $\inl(t)$.  This sequence can only reduce $t$, hence it is
finite.  Thus, $\inl(t)$ strongly terminates.

Furthermore, if $\inl(t) \longrightarrow^* \inl(t')$, then $t \lra^* t'$.
By Proposition~\ref{app-closure}, $t' \in \llbracket A
\rrbracket$. And $\inl(t)$ never reduces to $\inr(t')$. 
\end{proof}

\begin{proposition}[Adequacy of $\inr$]
\label{app-inr}
If $t \in \llbracket B \rrbracket$, then $\inr(t) \in \llbracket A
\vee B \rrbracket$.
\end{proposition}

\begin{proof}
  Similar to that of Proposition~\ref{app-inl}.  
\end{proof}

\begin{proposition}[Adequacy of $\elimtop$]
\label{app-elimtop}
If $t_1 \in \llbracket \top \rrbracket$ and $t_2 \in \llbracket C \rrbracket$, 
then $\elimtop(t_1,t_2) \in \llbracket C \rrbracket$.
\end{proposition}

\begin{proof}
The proofs $t_1$ and $t_2$ strongly terminate.  We prove, by
induction on $|t_1| + |t_2|$, that $\elimtop(t_1,t_2)
\in \llbracket C \rrbracket$.  Using Proposition~\ref{app-CR3}, we only
need to prove that every of its one step reducts is in $\llbracket C
\rrbracket$.  If the reduction takes place in $t_1$ or $t_2$, then we
apply Proposition~\ref{app-closure} and the induction hypothesis.

Otherwise, the proof $t_1$ is $\star$ and the
reduct is $t_2 \in \llbracket C \rrbracket$.
\end{proof}

\begin{proposition}[Adequacy of $\elimbot$]
\label{app-elimbot}
If $t \in \llbracket \bot \rrbracket$, 
then $\elimbot(t) \in \llbracket C \rrbracket$.
\end{proposition}

\begin{proof}
The proof $t$ strongly terminates.  Consider a reduction sequence
issued from $\elimbot(t)$.  This sequence can only reduce $t$, hence it
is finite.  Thus, $\elimbot(t)$ strongly terminates.  Moreover, it
never reduces to an introduction.
\end{proof}

\begin{proposition}[Adequacy of application]
\label{app-application}
If $t_1 \in \llbracket A \Rightarrow B \rrbracket$ and $t_2 \in
\llbracket A \rrbracket$, then $t_1~t_2 \in \llbracket B
\rrbracket$.
\end{proposition}

\begin{proof}
The proofs $t_1$ and $t_2$ strongly terminate. We prove, by induction
on $|t_1| + |t_2|$, that $t_1~t_2 \in \llbracket B \rrbracket$. Using
Proposition~\ref{app-CR3}, we only need to prove that every of its one
step reducts is in $\llbracket B \rrbracket$.  If the reduction takes
place in $t_1$ or in $t_2$, then we apply Proposition~\ref{app-closure}
and the induction hypothesis.

Otherwise, the proof $t_1$ has the form $\lambda \abstr{x}u$ and the reduct
is $(t_2/x)u$.  As $\lambda \abstr{x}u \in \llbracket A \Rightarrow B
\rrbracket$, we have $(t_2/x)u \in \llbracket B \rrbracket$.
\end{proof}

\begin{proposition}[Adequacy of $\elimand^1$]
\label{app-elimand1}
If $t_1 \in \llbracket A \wedge B \rrbracket$ and,
for all $u$ in $\llbracket A \rrbracket$,
$(u/x)t_2 \in \llbracket C \rrbracket $, 
then $\elimand^1(t_1, \abstr{x}t_2) \in \llbracket C \rrbracket$.
\end{proposition}

\begin{proof}
By Proposition~\ref{app-Var}, $x \in \llbracket A \rrbracket$
thus $t_2 = (x/x)t_2 \in \llbracket C
\rrbracket$.  Hence, $t_1$ and $t_2$ strongly terminate.  We prove, by
induction on $|t_1| + |t_2|$, that $\elimand^1(t_1, \abstr{x}t_2)
\in \llbracket C \rrbracket$.  Using Proposition~\ref{app-CR3}, we only
need to prove that every of its one step reducts is in $\llbracket C
\rrbracket$.  If the reduction takes place in $t_1$ or $t_2$, then we
apply Proposition~\ref{app-closure} and the induction hypothesis.

Otherwise, the proof $t_1$ has the form $\pair{u}{v}$ and the
reduct is $(u/x)t_2$.  As $\pair{u}{v} \in \llbracket A
\wedge B \rrbracket$, we have $u \in \llbracket A \rrbracket$.
Hence, $(u/x)t_2 \in \llbracket C \rrbracket$.
\end{proof}

\begin{proposition}[Adequacy of $\elimand^2$]
\label{app-elimand2}
If $t_1 \in \llbracket A \wedge B \rrbracket$ and,
for all $u$ in $\llbracket B \rrbracket$, 
$(u/x)t_2 \in \llbracket C \rrbracket $, 
then $\elimand^2(t_1,\abstr{x}t_2) \in \llbracket C \rrbracket$.
\end{proposition}

\begin{proof}
Similar to that of Proposition~\ref{app-elimand1}.
\end{proof}

\begin{proposition}[Adequacy of $\elimor$]
\label{app-elimor}
If $t_1 \in \llbracket A \vee B \rrbracket$, for all $u$ in $\llbracket A
\rrbracket$, $(u/x)t_2 \in \llbracket C \rrbracket $, and, for all $v$
in $\llbracket B \rrbracket$, $(v/y)t_3 \in \llbracket C \rrbracket $,
then $\elimor(t_1, \abstr{x}t_2, \abstr{y}t_3) \in \llbracket C \rrbracket$.
\end{proposition}

\begin{proof}
By Proposition~\ref{app-Var}, $x \in \llbracket A \rrbracket$, thus
$t_2 = (x/x)t_2 \in \llbracket C \rrbracket$. In the same way, $t_3
\in \llbracket C \rrbracket$.  Hence, $t_1$, $t_2$, and $t_3$ strongly
terminate.  We prove, by induction on $|t_1| + |t_2| + |t_3|$, that
$\elimor(t_1, \abstr{x}t_2, \abstr{y}t_3) \in \llbracket C
\rrbracket$.  Using Proposition~\ref{app-CR3}, we only need to prove
that every of its one step reducts is in $\llbracket C \rrbracket$.
If the reduction takes place in $t_1$, $t_2$, or $t_3$, then we apply
Proposition~\ref{app-closure} and the induction hypothesis. Otherwise,
either:
\begin{itemize}
\item The proof $t_1$ has the form $\inl(w_2)$ and the reduct is
  $(w_2/x)t_2$. As $\inl(w_2) \in \llbracket A \vee B \rrbracket$, we
  have $w_2 \in \llbracket A \rrbracket$.  Hence, $(w_2/x)t_2 \in
  \llbracket C \rrbracket$.

\item The proof $t_1$ has the form $\inr(w_3)$ and the reduct is
  $(w_3/x)t_3$. As $\inr(w_3) \in \llbracket A \vee B \rrbracket$, we
  have $w_3 \in \llbracket B \rrbracket$.  Hence, $(w_3/x)t_3 \in
  \llbracket C \rrbracket$.
\qedhere
\end{itemize}
\end{proof}

\begin{theorem}[Adequacy]
Let $t$ be a proof of $A$ in a context $\Gamma = x_1:A_1, \ldots, x_n:A_n$ and
$\sigma$ be a substitution mapping each variable $x_i$ to an element
of $\llbracket A_i \rrbracket$, then $\sigma t \in \llbracket A
\rrbracket$.
\end{theorem}

\begin{proof}
By induction on the structure of $t$.

\begin{itemize}
\item If $t$ is a variable, then, by definition of $\sigma$, $\sigma t
  \in \llbracket A \rrbracket$.

\item If $t = \star$, then, by Proposition~\ref{app-star}, $\star \in \llbracket
  \top \rrbracket$, that is $\sigma t \in \llbracket A \rrbracket$.

\item If $t = \lambda \abstr{x}u$, where $u$ is a proof of $C$, then, by
  induction hypothesis, for every $v \in \llbracket B \rrbracket$,
  $(v/x) \sigma u \in \llbracket C \rrbracket$. Hence, by
  Proposition~\ref{app-abstraction}, $\lambda \abstr{x}\sigma u \in \llbracket
  B
  \Rightarrow C \rrbracket$, that is $\sigma t \in \llbracket A
  \rrbracket$.

\item If $t = \pair{u}{v}$, where $u$ is a proof of $B$ and
  $v$ a proof of $C$, then, by induction hypothesis, $\sigma u \in
  \llbracket B \rrbracket$ and $\sigma v \in \llbracket C \rrbracket$.
  Hence, by Proposition~\ref{app-pair}, $\pair{\sigma u}{\sigma v}
  \in \llbracket B \wedge C \rrbracket$, that is $\sigma t \in
  \llbracket A \rrbracket$.

\item If $t = \inl(u)$, where $u$ is a proof of $B$, then, by induction
  hypothesis, $\sigma u \in \llbracket B \rrbracket$.  Hence, by
  Proposition~\ref{app-inl}, $\inl(\sigma u) \in \llbracket B \vee C
  \rrbracket$, that is $\sigma t \in \llbracket A \rrbracket$.

\item If $t = \inr(v)$, the proof is similar, using Proposition~\ref{app-inr}.

\item If $t = \elimtop(u,v)$, where $u$ is a proof of $\top$ and $v$
  is a proof of $A$, then, by induction hypothesis, $\sigma u \in
  \llbracket \top \rrbracket$, and $\sigma v \in \llbracket A
  \rrbracket$. Hence, by Proposition~\ref{app-elimtop}, $\elimtop(\sigma
  u, \sigma v) \in \llbracket A \rrbracket$, that is $\sigma t \in
  \llbracket A \rrbracket$.

\item If $t = \elimbot(u)$, where $u$ is a proof of $\bot$, then, by
  induction hypothesis, $\sigma u \in \llbracket \bot \rrbracket$.  Hence, by
  Proposition~\ref{app-elimbot}, $\elimbot(\sigma u) \in \llbracket A
  \rrbracket$, that is $\sigma t \in \llbracket A \rrbracket$.

\item If $t = u~v$, where $u$ is a proof of $B \Rightarrow A$ and $v$
  a proof of $B$, then, by induction hypothesis, $\sigma u \in
  \llbracket B \Rightarrow A \rrbracket$ and $\sigma v \in \llbracket
  B \rrbracket$.  Hence, by Proposition~\ref{app-application}, $(\sigma u)
  (\sigma v) \in \llbracket A \rrbracket$, that is $\sigma t \in
  \llbracket A \rrbracket$.

\item If $t = \elimand^1(u_1,\abstr{x}u_2)$, where $u_1$ is a
  proof of $B \wedge C$ and $u_2$ a proof of $A$, then, by
  induction hypothesis, $\sigma u_1 \in \llbracket B \wedge C
  \rrbracket$, for all $v$ in $\llbracket B \rrbracket$, $(v/x)\sigma u_2 \in
  \llbracket A \rrbracket$. Hence, by Proposition~\ref{app-elimand1},
  $\elimand^1(\sigma u_1,\abstr{x}\sigma u_2) \in \llbracket A
  \rrbracket$, that is $\sigma t \in \llbracket A \rrbracket$.

\item If $t = \elimand^2(u_1,\abstr{x}u_2)$, the proof is similar,
  using Proposition~\ref{app-elimand2}.
  
\item If $t = \elimor(u_1,\abstr{x}u_2,\abstr{y}u_3)$, where $u_1$
  is a proof of $B \vee C$, $u_2$ a proof of $A$, and $u_3$ a proof
  of $A$, then, by induction hypothesis, $\sigma u_1 \in \llbracket B
  \vee C \rrbracket$, for all $v$ in $\llbracket B \rrbracket$,
  $(v/x)\sigma u_2 \in \llbracket A \rrbracket$, and for all $w$ in
  $\llbracket C \rrbracket$, $(w/x)\sigma u_3 \in \llbracket A
  \rrbracket$. Hence, by Proposition~\ref{app-elimor}, $\elimor(\sigma
  u_1,\abstr{x} \sigma u_2,\abstr{y} \sigma u_3) \in \llbracket A
  \rrbracket$, that is $\sigma t \in \llbracket A \rrbracket$.
\qedhere
\end{itemize}
\end{proof}

\begin{corollary}[Termination]
Let $t$ be a proof of $A$ in a context $\Gamma$. Then, $t$ strongly
terminates.
\end{corollary}

\begin{proof}
Let $\sigma$ be the substitution mapping each variable $x_i:A_i$ of
$\Gamma$ to itself. Note that, by Proposition~\ref{app-Var}, this
variable is an element of $\llbracket A_i \rrbracket$.  Then, $t =
\sigma t$ is an element of $\llbracket A \rrbracket$. Hence, it
strongly terminates.
\end{proof}
\end{document}